\documentclass[11pt, reqno]{article}
\usepackage{algorithm}
\usepackage[noend]{algpseudocode}
\usepackage{amsmath, amssymb, amsthm}
\usepackage{thmtools}
\usepackage{thm-restate}
\usepackage{blindtext}
\usepackage{hyperref} % Must load hyperref before cleverref
\usepackage{cleveref}
\usepackage{comment}
\usepackage{enumerate}
\usepackage{fullpage}
\usepackage{mathtools}
\usepackage{subcaption}
\usepackage[dvipsnames]{xcolor}
\usepackage{dsfont}
\usepackage{xspace}

\hypersetup{
  colorlinks = true,
  linkcolor  = blue!80!black,
  citecolor  = blue!80!black,
  urlcolor   = blue!80!black,
}
\crefformat{equation}{(#2#1#3)}

\DeclareMathOperator*{\argmax}{arg\,max}
\DeclareMathOperator*{\argmin}{arg\,min}
\newcommand{\Prob}[2]{\Pr_{#1}\parens*{#2}}
\newcommand{\ProbCond}[3]{\Pr_{#1}\parens*{#2 \;\middle|\; #3}}
\newcommand{\E}{\mathbb{E}}
\newcommand{\Var}{\textnormal{Var}}
\newcommand{\Exp}[2]{\E_{#1}\bracks*{#2}}
\newcommand{\ExpCond}[3]{\E_{#1}\bracks*{#2 \;\middle|\; #3}}
\newcommand{\AdaptiveSampling}{\textsc{Threshold-Sampling}\xspace}
\newcommand{\AdaptiveSamplingForCover}{\textsc{Threshold-Sampling-For-Cover}\xspace}
\newcommand{\EstimateMean}{\textsc{Reduced-Mean}\xspace}
\newcommand{\Wrapper}{\textsc{Exhaustive-Maximization}\xspace}
\newcommand{\WrapperForCover}{\textsc{Adaptive-Greedy-Cover}\xspace}
\newcommand{\BinarySearch}{\textsc{Binary-Search-Maximization}\xspace}
\newcommand{\SubsamplePreprocess}{\textsc{Subsample-Preprocessing}\xspace}
\newcommand{\Subsample}{\textsc{Subsample-Maximization}\xspace}
\newcommand{\DEF}{\stackrel{\textnormal{\tiny\sffamily def}}{=}}

\newcommand{\R}{\mathbb{R}}
\newcommand{\OPT}{\textnormal{OPT}}
\newcommand{\poly}{\textnormal{poly}}

\newcommand{\ind}{\mathds{1}}
\newcommand{\change}[1]{{\color{black}#1}}
\newcommand{\morteza}[1]{{\color{black}#1}}

 \newcommand{\cD}{\mathcal{D}}

\DeclarePairedDelimiter{\abs}{\lvert}{\rvert}
\DeclarePairedDelimiter{\set}{\{}{\}}
\DeclarePairedDelimiter{\parens}{(}{)}
\DeclarePairedDelimiter{\bracks}{[}{]}
\DeclarePairedDelimiter{\floor}{\lfloor}{\rfloor}
\DeclarePairedDelimiter{\ceil}{\lceil}{\rceil}

\theoremstyle{plain}
\newtheorem{theorem}{Theorem}[section]
\newtheorem{lemma}[theorem]{Lemma}
\newtheorem{corollary}[theorem]{Corollary}

\theoremstyle{definition}
\newtheorem{definition}[theorem]{Definition}

\begin{document}

\title{Submodular Maximization with Nearly Optimal Approximation,
       Adaptivity and Query Complexity}

\author{
Matthew Fahrbach\thanks{
  School of Computer Science, Georgia Institute of Technology.
  Email: \href{mailto:matthew.fahrbach@gatech.edu}{\texttt{matthew.fahrbach@gatech.edu}}.
  Supported in part by a National Science Foundation Graduate Research
  Fellowship under grant DGE-1650044.
  Part of this work was done while the author was a summer intern at
  Google Research, Z\"urich.
}
\and
Vahab Mirrokni\thanks{
  Google~Research, New York. Email: \href{mailto:mirrokni@google.com}{\texttt{mirrokni@google.com}}.
}
\and
Morteza Zadimoghaddam\thanks{
  Google Research, Z\"urich. Email:
  \href{mailto:zadim@google.com}{\texttt{zadim@google.com}}.
}
}
\date{}
\maketitle

\begin{abstract}
Submodular optimization generalizes many classic problems in combinatorial
optimization and has recently found a wide range of applications in machine
learning (e.g., feature engineering and active learning).
For many
large-scale optimization problems, we are often concerned with the adaptivity
complexity of an algorithm, which quantifies the number of sequential rounds
where polynomially-many independent function evaluations can be executed in
parallel.  While low adaptivity is ideal, it is not sufficient for a
distributed algorithm to be efficient, since in many practical applications
of submodular optimization the number of function evaluations becomes
prohibitively expensive.  Motivated by these applications, we study the
adaptivity and query complexity of submodular optimization.

Our main result is a distributed algorithm for maximizing a monotone
submodular function with cardinality constraint $k$ that achieves 
a $(1-1/e-\varepsilon)$-approximation in expectation.
This algorithm runs in $O(\log(n))$
adaptive rounds and makes $O(n)$ calls
to the function evaluation oracle in expectation.
The approximation guarantee and query complexity are optimal,
and the adaptivity is nearly optimal.
Moreover, the number of queries is substantially less than in previous works.
We also extend our results to the submodular cover problem to
demonstrate the generality of our algorithm and techniques.
\end{abstract}

\pagenumbering{gobble}
\clearpage
\pagenumbering{arabic}

\section{Introduction}
\label{sec:introduction}

Submodular functions have the natural property of diminishing returns, making
them prominent in applied fields such as machine learning and data mining.
There has been a surge in applying submodular optimization for data
summarization \cite{tschiatschek2014learning,simon2007scene,sipos2012temporal},
recommendation systems \cite{el2011beyond}, and 
feature selection for learning models~\cite{DK08,KEDNG17},
to name a few applications.  There are also numerous recent works that focus on
maximizing submodular functions from a theoretical perspective. Depending on
the setting where the submodular maximization algorithms are applied, new
challenges emerge and hence more practical algorithms have been designed to
solve the problem in distributed \cite{nips13,MirrokniZadim2015,barbosa2015power},
streaming \cite{badanidiyuru2014streaming}, and robust
\cite{mirzasoleiman2017deletion,mitrovic2017streaming,kazemi2018scalable}
optimization frameworks.  Most of the existing work assumes access to an oracle that
evaluates the submodular function. However, function evaluations (oracle
queries) can take a long time to process---for example, the value
of a set depends on interactions with the entire input like in Exemplar-based
Clustering \cite{dueck2007non} or when the function is computationally hard to evaluate
like the log-determinant of sub-matrices \cite{kazemi2018scalable}.  Although
distributed algorithms for submodular maximization partition the input into
smaller pieces to overcome these problems, each distributed machine may run a
sequential algorithm that must wait for the answers of its past queries before
making its next query. This motivates the study of the adaptivity complexity of
submodular maximization, introduced by Balkanski and Singer~\cite{BS18} to
study the number of rounds needed to interact with the oracle. 
As long as we can ask
polynomially-many queries in parallel, we can ask them altogether in one round
of interaction with the oracle. 

To further motivate this adaptive optimization framework, note that in a wide
range of machine learning optimization problems, the objective function can
only be computed with oracle access to the function. \change{In settings where}
the \change{oracle computation} is a time-consuming optimization problem that
is treated as a black box (e.g., parameter tuning),
it is desirable to optimize a function with minimal number of
rounds of interaction with the oracle.
\change{For example, consider the feature selection problem~\cite{DK08,KEDNG17},
which is a critical step for improving the accuracy of machine learning models.}
The accuracy of a model \change{trained with a subset of features} does
not necessarily have a closed-form formula, and in \change{many settings must be evaluated
by re-training} the model \change{from scratch}.
\change{The training accuracy of certain models (e.g., generalized linear models) is known to be weakly submodular~\cite{DK08,KEDNG17}}.
In this case, we \change{only} have
black-box access to the model accuracy function, \change{and it can be} time-consuming to
compute.
However, the model accuracy of many \change{different feature subsets can be computed independently in parallel}. The adaptive optimization framework~\cite{BS18} is a
realistic model for \change{this type} of distributed problem,
and the insights from lower bounds and \change{algorithms} developed in
this framework have a deep impact on distributed computing for machine learning
applications in practice. For \change{further} motivation on the
importance of round complexity \change{in the} adaptive optimization framework,
we refer the reader to~\cite{BS18}.

While the number of rounds is an important \change{quantity} to optimize, the complexity
of answering oracle queries also motivates designing algorithms that are
efficient in terms of the total number of oracle queries. Typically, we need to
make at least a constant number of queries per element in the ground
set to have a constant approximation guarantee. A fundamental question is how many
queries per element are needed to achieve optimal approximation guarantees without
compromising the minimum number of adaptive rounds.
In this paper, we address this
issue and design a simple algorithm for \change{monotone} submodular maximization subject to
a cardinality constraint that achieves optimal guarantees for the approximation
factor and query complexity.
\change{Our algorithm also achieves nearly-optimal adaptivity complexity
using the lower bound in~\cite{BS18}.}

%\subsection{Our Results: Algorithms and Techniques}
\subsection{Results and Techniques}
Our main result is a simple distributed algorithm for maximizing a monotone
submodular function with cardinality constraint $k$ that achieves an
expected $(1-1/e-\varepsilon)$-approximation in $O(\log(n)/\varepsilon^2)$ adaptive
rounds and makes $O(n\log(1/\varepsilon)/\varepsilon^3)$ queries to the function
evaluation oracle in expectation.
We emphasize that while our algorithm runs in $O(\log(n)/\varepsilon^2)$
rounds, only a constant number of queries are made per element.  
We note that, due to known lower bounds~\cite{BS18,MBK15}, 
the query complexity of the algorithm is optimal up to factors of
$1/\varepsilon$ and the adaptivity is optimal up to factors of
$1/\log\log(n)$ and $1/\varepsilon$.
%the adaptivity and query complexity of our algorithms is 
%optimal up to factors of $1/\varepsilon$.
To achieve
this result, we develop a number of techniques and subroutines that can be
used in a variety of submodular optimization problems.

First, we develop the algorithm \AdaptiveSampling in \Cref{sec:threshold}, which
returns a subset of items from the ground set in $O(\log(n)/\varepsilon)$
adaptive rounds such that the expected marginal gain of each item in the
solution is at least the input threshold. Furthermore, upon terminating it
guarantees that all unselected items have marginal gain to the returned set
less than the threshold.
This effectively clears out all \change{high-value} items.
To achieve $O(\log(n)/\varepsilon)$ adaptivite complexity, \AdaptiveSampling
adds a random subset of candidate items to its current solution in each
round in such a way that probabilistically filters out an
$\varepsilon$-fraction of the remaining candidates.
%by dropping their marginal contribution below the fixed threshold.
We then use \AdaptiveSampling as a subroutine in a submodular maximization
\change{algorithm} that constructs a solution by gradually reducing its threshold for
acceptance.
This \change{algorithm} runs \AdaptiveSampling in parallel starting from many
different initial thresholds,
one of which is guaranteed to be sufficiently closed to the optimal starting
threshold. Consequently, we do not increase the adaptivity complexity because
these processes are independent.
One of the challenges that arises when analyzing the approximation factor of
this algorithm is that \AdaptiveSampling returns \change{a random set} of \change{(possibly)} variable size.
We overcome this by constructing an averaged random process that
agrees with the state of the maximization algorithm at the beginning and end,
but otherwise acts as an intermediate proxy.  In \Cref{sec:maximization}, we demonstrate how
to use \AdaptiveSampling as a subroutine in a greedy maximization \change{algorithm} to
achieve an expected $(1-1/e-\varepsilon)$-approximation to $\OPT$.

Our second main technical contribution is the \SubsamplePreprocess algorithm.
This algorithm iteratively subsamples the ground set and uses the output
guarantees of \AdaptiveSampling to reduce the ratio of the interval containing
$\OPT$ from $k$ to a constant.
The adaptivity complexity of this subroutine is $O(\log(n))$
and its query complexity is $O(n)$.
In particular, we show how to reduce the ratio of the interval
in each step from $R$ to $O(\poly(\log(R))$ by subsampling the ground set and
using a key lemma that relates $\OPT$ to the optimum in the subsampled set.
This approximation guarantee (\Cref{lem:subsample_approx}) for $\OPT$ is a
function of the subsampling rate and may be of independent interest. 
Our \change{ratio-reduction}
technique and the \SubsamplePreprocess algorithm are presented in \Cref{sec:linear-queries}.
Finally, in \Cref{sec:submodular-cover} we show how to use \AdaptiveSampling to
solve the submodular cover problem, \change{demonstrating} that our
techniques are readily applicable to problems beyond submodular maximization
\change{subject to a cardinality constraint.}

%\subsection{Further Related Work}
\subsection{Related Work}
The problem of optimizing query complexity for maximizing a submodular function
subject to cardinality constraints has been studied extensively.
In fact, a linear-time $(1-1/e-\varepsilon)$-approximation algorithm
called stochastic greedy was recently developed for this problem in
\cite{MBK15}.
We achieve the same optimal query complexity in this paper, combined with
nearly optimal $O(\log(n))$ adaptive round complexity.
The applications of efficient algorithms for submodular
maximization are \change{widespread} due to the numerous applications in machine
learning and data mining.
Submodular maximization has also recently attracted a significant amount of attention
in the streaming and distributed settings~\cite{spaa-LMSV11,KMVV13,nips13,badanidiyuru2014streaming,MirrokniZadim2015,barbosa2015power,alina2,CQ19}.
%The distributed MapReduce models are similar to the adaptive optimization model, since in both cases we are allowed to make parallel calls to the submodular function. 
We note that the distributed MapReduce model and adaptivity framework of
\cite{BS18} are different in that the latter
model does not allow for adaptivity within each round. 
In many previously studied distributed models, such as MapReduce, 
sequential algorithms on a given machine are allowed to be adaptive within one
round for the part of the data \change{that} they are processing locally.
To highlight the difference between these models, Balkanski and
Singer~\cite{BS18} showed that no constant-factor
approximation is achievable in $O(\log(n)/\log\log(n))$ non-adaptive rounds;
however, it is possible to achieve a constant-factor approximation in the
MapReduce model in two rounds~\cite{MirrokniZadim2015}. 

Balkanski and Singer~\cite{BS18} introduced the adaptive framework model for
submodular maximization and showed that a $(1/3)$-approximation is achievable
in $O(\log(n))$ rounds. Furthermore, they showed that
$\Omega(\log(n)/\log\log(n))$ rounds are
necessary for achieving any constant-factor approximation.
They left the problem of achieving the optimal approximation factor of $1-1/e$ open,
and as a followup posted a paper on arXiv achieving a
$(1-1/e-\varepsilon)$-approximation in $O(\log(n))$ rounds~\cite{BRS18}.
Their algorithm, however, requires $O(nk^2)$ queries~\cite{BRS18}.  While
writing this paper, another related  work (on arXiv) was brought to
our attention~\cite{EN18}. While \cite{EN18} has a similar goal to ours and aims
to minimize the number of adaptivity rounds and oracle queries, their
query complexity is $O(n \poly(\log(n)))$,
or $O(\poly(\log(n)))$ calls per element. In contrast,
we present a simple algorithm that achieves optimal query complexity
(i.e., a {\em constant number} of oracle queries per element). 
The query complexity of our algorithm is
optimal up to factors of $1/\varepsilon$.
While we did not aggresively optimize the dependence on $1/\varepsilon$, 
the dependence is better than that in the related works~\cite{BS18,BRS18,EN18}.

\section{Preliminaries}
\label{sec:preliminaries}
For a set function $f : 2^N \rightarrow \R$ and any $S,T \subseteq N$,
let $\Delta(T, S) \DEF f(S \cup T) - f(S)$ be the \emph{marginal gain} of $f$
at $T$ with respect to $S$.
We call $N$ the ground set and let $|N| = n$.
A function $f : 2^N \rightarrow \R$ is \emph{submodular} if for every
$S \subseteq T \subseteq N$ and $x \in N \setminus T$ we have
$\Delta(x, S) \ge \Delta(x, T)$, where we overload the
marginal gain notation for singletons.
A natural class of submodular functions are those which are \emph{monotone},
meaning that for every $S \subseteq T \subseteq N$ we have
$f(S) \le f(T)$.
In the inputs to our algorithms, we let $f_S(T) \DEF \Delta(T,S)$
denote a new submodular function with respect to $S$.
We also assume the ground set is global to all algorithms.
Let $S^*$ be a solution set to the 
maximization problem $\max_{S \subseteq N} f(S)$ subject to
the cardinality constraint $|S| \le k$.
\change{Lastly,} let $\mathcal{U}(A, t)$ denote the uniform distribution over all subsets of $A$
of size $t$. 

Our algorithms take as input an \emph{evaluation oracle} for $f$, which for any
query $S \subseteq N$ returns the value of $f(S)$ in $O(1)$ time.
Given an evaluation oracle, we define the \emph{adaptivity} of
an algorithm to be the minimum number of rounds such that in each
round the algorithm can make polynomially-many independent queries to the
evaluation oracle.
We measure the complexity of our distributed algorithms in
terms of their query and adaptivity complexity.
\change{Finally, we note} that in our runtime guarantees, we
take $1/\delta = \Omega(\poly(n))$
so the claims hold with high probability.

\section{\AdaptiveSampling Algorithm}
\label{sec:threshold}

We start by giving a high-level description of the \AdaptiveSampling algorithm.
For an input threshold $\tau$, the algorithm iteratively builds a solution $S$
and maintains a set of unselected candidate elements~$A$ over
$O(\log(n)/\varepsilon)$ adaptive rounds.  Initially, the solution is empty and
all elements are candidates.
In each round, the algorithm starts by filtering out candidiate elements whose
current marginal gain is less than the threshold.
Then the algorithm efficiently finds the largest set size~$\change{t^* = \min\{\floor{(1 + \hat{\varepsilon})^i}, |A|\}}$, \change{where $i \in \mathbb{Z}_{\ge 0}$ and $\hat{\varepsilon}$ is a small fixed parameter defined in the algorithm,}
such that for
$T \sim U(A, t^*)$ uniformly at random we have
$\E[\Delta(T,S)/|T|] \ge (1-\varepsilon)\tau$.
Next, the algorithm samples $T \sim U(A, t^*)$ and updates the current solution
to $S \cup T$.
This probabilistic guarantee has two beneficial effects:
\change{
\begin{enumerate}
    \item It ensures the average contribution of each element in the output set $S$ is at least $(1-\varepsilon)\tau$.
    \item It implies that an $\varepsilon$-fraction of candidates are filtered out of $A$ in each round in expectation.
\end{enumerate}
}
\noindent
Therefore, the number of remaining elements that the algorithm considers in
each round decreases geometrically in expectation. It follows that
$O(\log(n)/\varepsilon)$ rounds are sufficient to guarantee with high
probability that when the algorithm terminates, we \change{either} have $|S| = k$ or
the marginal gain of all remaining elements \change{to $S$} is below the threshold.

Before presenting \change{the \AdaptiveSampling algorithm}, we define the \change{probability} distribution
$\cD_t$ from which \AdaptiveSampling~\change{draws} samples when estimating the maximum set
size $t^*$ in each round. Sampling from $\cD_t$ can be
\change{implemented} with two calls to the evaluation oracle.

\begin{definition}
\label{def:indicator_distribution}
Conditioned on the current state of the algorithm,
consider the random process where
\change{we sample
$T \sim \mathcal{U}(A, t)$} and then
$x \sim A\setminus T$ uniformly at random.
\change{For all $t \in \{0,1,\dots,|A|-1\}$,}
let $\mathcal{D}_t$ denote the \change{Bernoulli} distribution \change{for} the indicator random variable
\begin{align*}
  I_t~\change{\DEF}~\ind\bracks*{\Delta(x,S\cup T) \ge \tau}.
\end{align*}
\change{For completeness, we define $I_{|A|} = 0$.}
\end{definition}

\noindent
It is useful to think of $\Exp{}{I_t}$ as the probability that the \change{$(t+1)$-st
marginal gain} is at least threshold~$\tau$ if the candidates in $A$ are inserted
into $S$ according to a uniformly random permutation.

Now that $\cD_t$ is defined, we present
the \AdaptiveSampling algorithm and its guarantees below.
Observe that this algorithm calls the \EstimateMean subroutine, which
detects when the mean of $\cD_t$ falls below $1 - \varepsilon$.
We give the exact guarantees of \EstimateMean in \Cref{lem:estimator}.
Relating the mean of $\cD_t$ to threshold values, this means that
after sampling $T \sim U(A,t^*)$ and adding the elements of $T$ to $S$,
the expected marginal gain of the remaining candidates to $S$ is at most
$(1-\varepsilon)\tau$.
This is the invariant we want to maintain in each iteration for an
$O(\log(n/\delta)/\varepsilon)$ adaptive algorithm.
We explain the mechanics of \AdaptiveSampling in detail
and prove \Cref{lem:adaptive_sampling} in
\Cref{sec:analysis_adaptive_sampling}.

\begin{algorithm}[H]
  \caption{\AdaptiveSampling}
  \label{alg:sampling}
  \vspace{0.1cm}
  \textbf{Input:} evaluation oracle for $f : 2^N \rightarrow \change{\R_{\ge 0}}$,
                  constraint $k$, threshold $\tau$,
    error $\varepsilon$, failure probability $\delta$
  \begin{algorithmic}[1]
    \State Set smaller error $\hat{\varepsilon} \leftarrow \varepsilon/3$
    \State Set iteration bounds
        $r \leftarrow \ceil{\log_{(1-\hat\varepsilon)^{-1}}(2n/\delta)}$,
        $m \leftarrow \ceil{\log_{\change{(1+\hat\varepsilon)}}(k)}$
    \State Set smaller failure probability
        $\hat{\delta} \leftarrow \delta/(2r(m+1))$
    \State Initialize $S \leftarrow \emptyset$, $A \leftarrow N$
    \For{$r$ rounds}
      \State Filter $A \leftarrow \{x \in A : \Delta(x,S) \ge \tau\}$
      \If{$\abs{A} = 0$}
        \State \textbf{break}
      \EndIf
      \For{$i=0$ to $m$} 
        \State Set $t \leftarrow \min\{\floor{(1 + \hat{\varepsilon})^i}, |A|\}$
        \If{$\EstimateMean(\mathcal{D}_t, \hat{\varepsilon},
          \hat{\delta} )$} %\Comment{\Cref{def:indicator_distribution}}
          \State \textbf{break}
        \EndIf
      \EndFor
      \State Sample $T \sim \mathcal{U}(A, \min\set{t, k - \abs{S}})$
      \State Update $S \leftarrow S \cup T$
      \If{$\abs{S} = k$}
        \State \textbf{break}
      \EndIf
    \EndFor
    \State \textbf{return} $S$
  \end{algorithmic}
\end{algorithm} 

\begin{restatable}[]{lemma}{adaptiveSampling}\label{lem:adaptive_sampling}
\change{Let $Z$ be the event that all calls to \EstimateMean give correct outputs
(i.e., the reported property in \Cref{lem:estimator} holds).}
\change{For any monotone,\footnote{\change{In the subsequent work~\cite{fahrbach2019non},
the authors modify~\AdaptiveSampling to give guarantees for non-monotone submodular functions.}}
nonnegative submodular function~$f$,}
\AdaptiveSampling outputs $S \subseteq N$ with $|S| \le k$ in
\change{$O(\log(n/\delta)/\log(1/(1-\varepsilon)))$ adaptive rounds (for small $\varepsilon$, this becomes $O(\log(n/\delta)/\varepsilon)$)} such that
the following properties hold conditioned on $Z$:
  \begin{enumerate}
    \item \change{The algorithm makes} $O(n/\varepsilon)$ oracle queries in expectation.
    \item The \change{average marginal gain satisfies} $\Exp{}{f(S)/|S|} \ge (1-\varepsilon)\tau$.
    \item
    With probability at least $1-\delta/2$,
    if $|S| < k$,
    then $\Delta(x, S) < \tau$ for all $x \in N$.
  \end{enumerate}
\change{Further, event $Z$ happens with probability at least $1 - \delta/2$.}
\end{restatable}

\begin{algorithm}[H]
  \caption{\EstimateMean}
  \label{alg:estimate_mean}
  \vspace{0.1cm}
  \textbf{Input:} Bernoulli distribution $\mathcal{D}$,
    error $\varepsilon$, failure probability $\delta$
  \begin{algorithmic}[1]
    \State Set number of samples
      $m \leftarrow 16 \ceil{\log(2/\delta)/\varepsilon^2}$
    \State Sample $X_1, X_2,\dots, X_m \sim \mathcal{D}$
    \State Set $\overline{\mu} \leftarrow \frac{1}{m} \sum_{i=1}^m X_i$
    \If{$\overline{\mu} \le 1 - 1.5\varepsilon$}
      \State \textbf{return} \texttt{true}
    \EndIf
    \State \textbf{return} \texttt{false}
  \end{algorithmic}
\end{algorithm} 

\begin{restatable}[]{lemma}{estimator}\label{lem:estimator}
For any Bernoulli distribution $\mathcal{D}$,
\EstimateMean uses $O(\log(\delta^{-1})/\varepsilon^{2})$
samples to \change{report one of the following properties,
which is correct with probability at least $1-\delta$:}
\begin{enumerate}
  \item If the output is \textnormal{\texttt{true}}, then the mean of $\mathcal{D}$ is $\mu \le 1 - \varepsilon$.
  \item If the output is \textnormal{\texttt{false}}, then the mean of $\mathcal{D}$ is $\mu \ge 1 - 2\varepsilon$.
\end{enumerate}
\end{restatable}

We briefly remark that the \EstimateMean subroutine is a standard unbiased
estimator for the mean of a Bernoulli distribution.
Since $\cD_t$ is a uniform distribution over indicator random
variables, it is a Bernoulli distribution.
The guarantees of in \Cref{lem:estimator} are consequences of Chernoff bounds
and the proof of \Cref{lem:estimator} is given
in \Cref{app:estimator_analysis}.

\subsection{Analysis of \AdaptiveSampling Algorithm}
\label{sec:analysis_adaptive_sampling}

To prove the guarantees of \AdaptiveSampling~\change{in \Cref{lem:adaptive_sampling}},
we first \change{show that $\cD_t$ has monotonic behavior} at any
point in the algorithm.
This is a simple consequence of submodularity, and
the proof can be found in \Cref{app:nonincreasing_proof}.

\begin{restatable}[]{lemma}{percentageNonincreasing}
\label{lem:percentage_nonincreasing}
In each round of \AdaptiveSampling,
we have $\change{1} = \E[I_\change{0}] \ge \E[I_\change{1}] \ge \dots \ge \E[I_{|A|}]$.
\end{restatable}

Now we show that if we choose the maximum set size $t^*$ in each round such
that the average marginal gain of a randomly sampled subset of size $t^*$ is at
least $(1-\varepsilon)\tau$, then we expect to filter an $\varepsilon$-fraction
of the remaining candidates in the subsequent round.
In \Cref{lem:filters-all} we show that our choice of the number of rounds
is sufficient to guarantee that
all unchosen elements have marginal gain less than $\tau$ with high probability.

\begin{lemma}
\label{lem:filter}
% OLD:
% In each round of \AdaptiveSampling, an expected $\hat{\varepsilon}$-fraction
% of $A$ is filtered with probability at least $1 - \hat{\delta}$.
\change{Conditioned on event $Z$,
in each round of \AdaptiveSampling,
we expect to filter out an $\hat{\varepsilon}$-fraction of the elements in $A$.}
\end{lemma}

\begin{proof}
This is a consequence of our choice of
$t^* = \min\{t, k-|S|\}$ when sampling $T$.
If $t^* = k - |S|$ \change{or $t^* = |A|$,} then the algorithm breaks from the loop and there is
no subsequent filtering.
Otherwise, for any given round,
we condition on the state of the algorithm.
Let $A_i$ \change{be} the value of $A$ after the filtering step in the \change{$i$-th}
round, and let $A_{i+1}$ be the random variable for the future value of $A_i$
after being filtered in the next round.
The algorithm draws $T \sim \mathcal{U}(A_{i},t^*)$
uniformly at random, so by the process in
\Cref{def:indicator_distribution},
filtering has the property that for $x \sim A_{i} \setminus T$,
\begin{align*}
  \Exp{}{I_{\change{t^*}}}
  &= \Prob{}{\Delta(x, S\cup T) \ge \tau}
  = \Exp{}{\frac{\abs*{A_{i+1}}}{\abs*{A_{i} \setminus T}}}.
\end{align*}
\change{We have $\E\bracks{I_{t^*}} \le 1 - \hat{\varepsilon}$} by \Cref{lem:estimator}.
%Using \Cref{lem:percentage_nonincreasing} gives us
%$\E\bracks{I_{t^* + 1}} \le \E\bracks{I_{t^*}} \le 1 - \hat{\varepsilon}$.
\change{Since} $|A_{i} \setminus T| \le |A_{i}|$ for all choices of $T$,
\change{it follows that}
\[
  \Exp{}{\abs*{A_{i+1}}} 
  \le \parens*{1-\hat\varepsilon} \cdot \abs*{A_{i}}.
\]
\change{Therefore,}
an expected $\hat\varepsilon$-fraction of elements are filtered out in each round.
\end{proof}

\begin{lemma}
\label{lem:filters-all}
% OLD
%If \AdaptiveSampling terminates with $|S| < k$, 
%then $|A| = 0$ with probability at least $1 - \delta$.
\change{Conditioned on event $Z$,
if \AdaptiveSampling terminates with $|S| < k$, then
we have $|A| = 0$ with probability at least $1 - \delta/2$.
}
\end{lemma}

\begin{proof}
Denote by $A_i$ the random variable for the value of $A$ after it is filtered
in the $i$-th round of \AdaptiveSampling, and
\change{recall} that $A_0 = N$.
\Cref{lem:filter} \change{gives us $\E[|A_{i+1}|] \le (1-\hat\varepsilon) \cdot  \E[|A_{i}|]$,
which implies}
\[
  \Exp{}{\abs*{A_{r}}} \le \parens*{1 - \hat\varepsilon}^r \cdot \Exp{}{\abs*{A_{0}}}
  = \parens*{1 - \hat\varepsilon}^r n.
\]
Therefore, by Markov's inequality and our choice of the number of
rounds $r$, we have
\begin{align*}
  \Prob{}{\abs*{A_r} \ge 1}
  \le \parens*{1 - \hat\varepsilon}^{\log_{(1-\hat\varepsilon)^{-1}}(2n/\delta)} n
  = \delta/2.
\end{align*}
It follows that
$\Prob{}{|A_{r}| = 0} = 1 - \Prob{}{|A_{r}| \ge 1} \ge 1 - \delta/2$,
\change{as desired, conditioned on event $Z$.}
\end{proof}

Using the guarantees for \EstimateMean and
the two lemmas above, we can prove \Cref{lem:adaptive_sampling}.

\begin{proof}[Proof of \Cref{lem:adaptive_sampling}]
We start by showing that the adaptivity complexity of \AdaptiveSampling is
\change{always} $O(\log(n/\delta)/\varepsilon)$.
By construction,
the number of rounds is $O(\log_{(1-\varepsilon)^{-1}}(n/\delta))$,
\change{and in each round the algorithm makes polynomially-many queries},
all of which are independent
and \change{only} rely on the current state of $S$.

% To prove the three properties, we use \Cref{lem:filters-all} to assume
% that with probability at least $1 - \delta$ all $O(rm)$ calls to \EstimateMean yield
% correct outputs, and also that if the algorithm terminates with $|S| < k$ then
% we have $|A| = 0$.
\change{Assume that event $Z$ is true.}
For Property 1, the total number of oracle queries incurred from calling
\EstimateMean is
  $O\parens{rm\log\parens*{\delta^{-1}}/\varepsilon^{2}}
  = 
  O\parens{\log\parens*{n/\delta}\log\parens*{k}\log\parens*{\delta^{-1}}/\varepsilon^4}$
by \Cref{lem:estimator}.
Note that we can sample from $\mathcal{D}_t$ with two oracle calls.
Now we bound the \change{total} expected number of queries made while filtering over
the course of the algorithm. Let~$A_i$ be a random variable for the value
of $A$ in the $i$-th round.
It follows from the \change{inequality}
$\E[|A_{i+1}|] \le (1 - \hat\varepsilon) \cdot \E[|A_i|]$ in
the proof of \Cref{lem:filter} and the linearity of expectation
that the expected number of queries is bounded by
\begin{align*}
  \Exp{}{\sum_{i=0}^{\change{r-1}} \abs*{A_{i}}} = \sum_{i=0}^{\change{r-1}} \Exp{}{\abs*{A_i}}
  \le n \sum_{i=0}^{\change{r-1}} \parens*{1 - \hat\varepsilon}^i
  \le n/\hat\varepsilon.
\end{align*}
Since we set $\delta^{-1} = O(\poly(n))$, the number of expected
queries made when filtering dominates the sum of queries made when calling
\EstimateMean.

For \change{P}roperty 2,
it suffices to lower bound the expected marginal \change{gain}
of every element added to~$S$
if we think of adding each set $T$ to the output $S$
one element at a time according to a uniformly random permutation.
Let $t^* = \min\{t, k-|S|\}$ be the size of $T$ \change{in} an arbitrary
round.
If $t^* = 1$, then $\E[\Delta(T,S)] \ge \tau$ by the definition of $A$.
Otherwise, the size \change{$t \le t^*$ in the previous step, where
$t = \floor{(1+\hat\varepsilon)^{i-1}} \le \floor{(1+\hat\varepsilon)^{i}} = t^*$,}
satisfies $\E[I_t] \ge 1 - 2\hat\varepsilon$ \change{since the algorithm
did not break on Line~11.}
%\change{Note that we have $t^* \le (1+\hat\varepsilon)t$.}
\change{Then,}
since $T \sim \mathcal{U}(A,t^*)$ is sampled uniformly at random,
we can lower bound $\E[\Delta(T,S)]$ by the contribution of the
first \change{$\min\{t+1,t^*\}$} elements, giving us
\begin{align}
  \Exp{}{\Delta(T,S)} &\ge
  \parens*{\Exp{}{I_\change{0}} + \Exp{}{I_\change{1}} + \dots + \Exp{}{I_{\change{t^*-1}}}}\tau \label{eqn:property_2_ineq_1}\\
  &\ge \change{(\min\{t+1,t^*\})} (1-2\hat\varepsilon) \tau \label{eqn:property_2_ineq_2}\\
  &\ge \frac{t^*}{1 + \hat\varepsilon} \cdot (1-2\hat\varepsilon) \tau \label{eqn:property_2_ineq_3}\\
  &\ge t^* (1-\varepsilon) \tau. \notag
\end{align}
\change{Inequality~\Cref{eqn:property_2_ineq_1}}
uses the definition of $I_t$
and is \change{essentially} Markov's inequality.
%\change{\footnote{
%\todo{Point out that this is where the monotone assumption gets used.}}}
\change{Note that \Cref{eqn:property_2_ineq_1} focuses only on the elements for which we achieve a marginal gain of at least $\tau$.
For the rest of the elements in $T$, we use the monotonicity of function $f$
and argue that their marginal gains are nonnegative.}
\change{Inequality~\Cref{eqn:property_2_ineq_2} uses \Cref{lem:percentage_nonincreasing}
and the fact that $\E[I_t] \ge 1 - 2\hat\varepsilon$.
Inequality~\Cref{eqn:property_2_ineq_3}
uses the observation that $t^* \le (1+\hat\varepsilon)(t+1)$.}
\change{Thus, Property 2 follows since the expected marginal gain
of any individual element in $T$ is at least $(1-\varepsilon)\tau$.}

\change{Property 3 follows from
\Cref{lem:filters-all}, the definition of $A$, and submodularity.}

\change{
Finally, it remains to show that event $Z$ happens with probability at least $1-\delta/2$.
Recall that each call to \EstimateMean succeeds with probability at least
$1 - \hat{\delta}$ by \Cref{lem:estimator}.
Therefore, by a union bound, all $r(m+1)$
calls succeed with probability at least 
$1-r(m+1)\hat{\delta} = 1 - \delta/2$.
}
\end{proof}

\vspace{-0.33cm}
\section{\Wrapper Algorithm}
\label{sec:maximization}

In this section we show how \AdaptiveSampling fits into a greedy framework
for maximizing monotone submodular functions \change{subject to} a cardinality constraint.
We start by presenting the \Wrapper algorithm, and then we prove its
guarantees in \Cref{sec:wrapper_analysis}.
The algorithm \Wrapper works as follows.
Given an initial threshold~$\tau$, \Wrapper constructs a solution
by repeatedly running \AdaptiveSampling 
at decreasing thresholds $(1-\varepsilon)^j \tau$ conditioned on the
current partial solution.
It is greedy in the sense that every time \AdaptiveSampling is called,
the expected average contribution of the elements in the returned set
is at least $(1-\varepsilon)\tau$ and
the marginal gain of all remaining elements
is less than the current threshold.
These properties allow us to prove an approximation guarantee with respect to
the initial threshold~$\tau$.

To relate the quality of the solution to $\OPT$, we first let
$\Delta^* = \max\{f(x) : x \in N\}$ be an upper bound for all marginal
contributions by submodularity and
observe that $\Delta^* \le \OPT \le k\Delta^*$.
The threshold that \Wrapper searches for is $\tau^* = \OPT/k$, so
it suffices to run the greedy thresholding algorithm for $O(\log(k)/\varepsilon)$
initial thresholds $(1+\varepsilon)^i \Delta^*/k$ in parallel
and return the solution with maximum value.
Since the algorithm will try some threshold $\tau$
close enough to~$\tau^*$, specifically $\tau \le \tau^* \le (1+\varepsilon)\tau$,
the approximation to $\OPT$ follows.
Note that by trying all thresholds in parallel, the adaptivity complexity
of the algorithm does not increase.
In \Cref{sec:linear-queries} we present efficient preprocessing methods
\change{to reduce} the ratio of the interval containing $\OPT$.

\begin{algorithm}[H]
  \caption{\Wrapper}
  \label{alg:wrapper}
  %\vspace{0.1cm}
  \textbf{Input:} evaluation oracle for $f : 2^N \rightarrow \change{\R_{\ge 0}}$, constraint $k$, error $\varepsilon$, failure probability $\delta$
  \begin{algorithmic}[1]
%    \State \todo{Set smaller error: $\hat{\varepsilon} \leftarrow \varepsilon/2$? Propagate change.}
    \State Set upper bounds $\Delta^* \leftarrow \argmax\{f(x) : x \in N\}$,
    $r \leftarrow \ceil{2\log(k)/\varepsilon}$,
    $m \leftarrow \ceil{\log(4)/\varepsilon}$
    \State Set smaller failure probability $\hat{\delta} \leftarrow \delta/(r(m+1))$
    \State Initialize $R \leftarrow \emptyset$
    \For{$i=0$ to $r$ in parallel}
      \State Set $\tau \leftarrow (1+\varepsilon)^i \Delta^*/k$
      \State Initialize $S \leftarrow \emptyset$
      \For{$j=0$ to $m$} \Comment{Until $(1-\varepsilon)^j\tau < \tau/4$}
        \State Set $T \leftarrow
    \AdaptiveSampling(f_S$,
    ${k - |S|}$, $(1-\varepsilon)^j \tau, \varepsilon, \hat{\delta})$
        \State Update $S \leftarrow S \cup T$
        \If{$|S|=k$}
          \State \textbf{break}
        \EndIf
      \EndFor
    \If{$f(S) > f(R)$}
      \State Update $R \leftarrow S$
    \EndIf
    \EndFor
    \State \textbf{return} $R$
  \end{algorithmic}
\end{algorithm} 

\begin{restatable}[]{theorem}{approx}\label{thm:approx_factor}
For \change{a} monotone, nonnegative submodular function $f$,
\Wrapper outputs a set $S \subseteq N$ with $|S| \le k$ in
$O(\log(n/\delta)/\varepsilon^{2})$ adaptive rounds 
\change{and with $O(n \log(k)/\varepsilon^3 + \delta n^2)$ oracle queries in expectation 
such that $\E[f(S)] \ge (1 - 1/e - \varepsilon)(1-\delta)\OPT$.
Setting $\delta < 1/n$ yields a good trade-off 
between the number of adaptive rounds  and oracle calls.}
\end{restatable}

\subsection{Analysis of \Wrapper Algorithm}
\label{sec:wrapper_analysis}

To analyze the expected approximation factor of \Wrapper, we first assume
that all calls to \AdaptiveSampling give correct outputs by our choice of
$\hat{\delta}$ and a union bound
\change{(i.e., event $Z$ in \Cref{lem:adaptive_sampling} always holds).}
The analysis is for one execution of the block in the
for loop of \Wrapper (Line~5 to Line 13),
\change{and it} assumes the initial value of $\tau$ is
sufficiently close to $\tau^* = \OPT/k$, satisfying $\tau \le \tau^* \le (1+\varepsilon)\tau$.
Furthermore, we assume that
the final output set is of size~$k$.
We refer to this modified block of \Wrapper as the algorithm.

For a fixed input, as the algorithm runs it produces nonempty sets
$T_1,T_2,\dots,T_m$, inducing a probability distribution over
sequences of subsets.
Denote their respective sizes by $t_1,t_2,\dots,t_m$ and the input values
$(1-\varepsilon)^j \tau$
for which they were returned by $\tau_1, \tau_2,\dots,\tau_m$.
We view the algorithm
as a random process that adds elements to the output set $S$ one at a time
instead of set by set. Specifically, for each new $T_i$ the algorithm adds each $x \in
T_i$ to $S$ in lexicographic order, producing a sequence of subsets
$S_0 \subseteq S_1 \subseteq \dots \subseteq S_k$ with $S_0 = \emptyset$.
Note that there is no randomness
in adding the elements of $T_i$ once $T_i$ is drawn.

Instead of analyzing the expected value of the partial solution $\E[f(S_i)]$ at
each step, we consider an averaged version of this random process that is
easier to analyze and whose final expected value is equal to $\E[f(S_k)]$.
In particular, when the process draws a set $T_\ell$,
each element $x \in T_\ell$ contributes the same amount
$\Delta(T_\ell, T_1 \cup \dots \cup T_{\ell-1})/|T_{\ell}|$ to the value of the
output set.
Formally, we define the averaged version of the random process as
\begin{align*}
  \hat{f}(S_i) &\DEF
  f(T_1 \cup \dots \cup T_{\ell-1})
  + 
  \frac{i - (t_1 + \dots + t_{\ell-1})}{t_\ell}
  \cdot
  \Delta(T_\ell, T_1\cup \dots \cup T_{\ell-1}),
\end{align*}
where we use the overloaded notation
$S_i = (i, T_1, T_2, \dots, T_\ell)$ to record
the history of the process up to adding the $i$-th element.
This means that
$t_1 + t_2 + \dots + t_{\ell - 1} < i$ and
$t_1 + t_2 + \dots + t_{\ell} \ge i$.
Note that for a given history $T_1,T_2,\dots,T_m$ of subsets, both processes
agree after adding a complete subset.
Analogously, we define the marginal of the $i$-th element $X_i$ of this
process to be
\[
  \hat\Delta(X_i, S_{i-1}) \DEF
  \frac{\Delta(T_\ell, T_1\cup \dots \cup T_{\ell-1})}{t_\ell}.
\]
Since the original algorithm induces a probability distribution over sequences of
returned subsets, this defines a distribution over the
values of $\hat{f}(S_i)$ and $\hat{\Delta}(X_i, S_{i-1})$ for all indices $i \in [k]$.

Lastly, it will be useful to define the distribution $\mathcal{H}$ over all
possible (random bit) histories $(T_1,T_2,\dots,T_m)$
at the termination of the algorithm, and also the distributions
$\mathcal{H}_i$, for all $i \in [k]$, over the
possible histories \emph{immediately} before adding the $i$-th
element. This means that for each $h = (T_1,T_2,\dots,T_{\ell-1}) \in
\text{supp}(\mathcal{H}_i)$ we have
$t_1 + t_2 + \dots + t_{\ell-1} < i$ and
there exists a set~$T_\ell$ that can be drawn such that
$t_1 + t_2 + \dots + t_{\ell} \ge i$.
Let $H_i(h)$ be the event over $\mathcal{H}_i$ such that the history is $h =
(T_1,T_2,\dots,T_{\ell-1})$ and the next returned subset $T_\ell$ adds the
$i$-th element.
We condition our statements on $H_i(h)$, as this captures
the state of the algorithm just before adding the $i$-th element.
To provide intuition for \Cref{lem:expectation_approx},
it is worth noting that $\mathcal{H}$ is a refinement of
$\mathcal{H}_i$ conditioned on $H_i(h)$. This can be seen
by recursively joining leaves in the probability tree of $\mathcal{H}$
until the $i$-th element is reached. The
result is the probability tree of $\mathcal{H}_i$
conditioned on $H_i(h)$.
In the statements that follow, the probabilities and expectations
conditioned on $H_i(h)$ are over the distribution $\mathcal{H}_i$ and
all other expressions are over the distribution $\mathcal{H}$ of final
outcomes.

\begin{lemma} 
\label{lem:conditional_iterative_property}
%For all $i \in [k]$, let $H_i$  be the event that (random bit) history
%is $h = (T_1,\dots,T_{\ell-1})$ and the next returned subset $T_\ell$ adds the
%$i$-element (i.e., $t_1 + \dots + t_{\ell-1} < i$ and $t_1 + \dots + t_\ell \ge i$).
%For $\tau \le \tau^* \le (1+\varepsilon)\tau$, we have
\change{Conditioned on event $Z$ (defined in \Cref{lem:adaptive_sampling}),}
for any $i \in [k]$, event $H_i(h)$, and threshold $\tau$ such that
$\tau \le \tau^* \le (1+\varepsilon)\tau$, we have
\begin{align*}
  \ExpCond{}{\hat{\Delta}\parens*{X_i, S_{i-1}}}{H_i(h)}
  \ge
  \frac{\parens*{1-\varepsilon}^2}{k}\cdot\ExpCond{}{\OPT - \hat{f}(S_{i-1})}{H_i(h)}.
\end{align*}
\end{lemma}

\begin{proof}
%Assume that the initial approximation $\tau$ to $\tau^* = \OPT/k$ satisfies
%$\tau^* \le \tau \le (1+\varepsilon)\tau^*$.
First we prove the claim for $i = 1$ and then we proceed by case analysis.
If $i=1$ there is no history, so it suffices to show that
  $E[\hat{f}(S_1)] \ge (1-\varepsilon)^2\OPT/k$.
The first element belongs to the subset $T_1$ returned by
\AdaptiveSampling, so by
Property 2 of \Cref{lem:adaptive_sampling},
it follows that
\begin{align*}
  \Exp{}{\hat{f}\parens*{S_1}}
  &=
  \Exp{}{\frac{f\parens*{T_1}}{\abs*{T_1}}}
  \ge \parens*{1-\varepsilon} \tau
  %&\ge \parens*{\frac{1-\varepsilon}{1+\varepsilon}} \cdot \frac{\OPT}{k}\\
  \ge \parens*{1-\varepsilon}^2 \cdot \frac{\OPT}{k}.
\end{align*}

Assuming that $i > 1$,
  let $i^* = {t_1 + \dots + t_{\ell - 1} + 1}$ be the size of the partial solution
after adding the first element in $T_\ell$.
We consider the cases $i=i^*$ and $i > i^*$ separately.
If $i = i^*$, observe that for monotone submodular functions we have
\begin{align*}
  f\parens*{S^*} &\le f\parens*{S^* \cup S_{i-1}}\\
  &\le f\parens*{S_{i-1}} + \sum_{x \in S^*} \Delta(x, S_{i-1})\\
  &\le f(S_{i-1}) + k \cdot \change{\tau_{\ell-1}} & \text{(Property 3 of \Cref{lem:adaptive_sampling})}\\
  &= f(S_{i-1}) + k \cdot \change{\frac{\tau_{\ell}}{1-\varepsilon}} \\
  &\le f(S_{i-1}) + \frac{k}{(1-\varepsilon)^2} \cdot \ExpCond{}{\frac{\Delta\parens*{T_\ell, S_{i-1}}}{\abs*{T_\ell}}}{H_i(h)} & \text{(Property 2 of \Cref{lem:adaptive_sampling})}\\
  &= f(S_{i-1}) + \frac{k}{(1-\varepsilon)^2} \cdot \ExpCond{}{\hat{\Delta}\parens*{X_i, S_{i-1}}}{H_i(h)}.
\end{align*}
In the \change{fourth line},
we have $\tau_{\ell}/(1-\varepsilon)$ 
because \AdaptiveSampling was run with parameter \change{$\tau_{\ell-1}$}
immediately before running with threshold $\tau_{\ell}$, which returned $T_\ell$.
The upper bound for the marginal \change{gain} $\Delta(x, S_{i-1})$ for $x \in
S^*$ is then a consequence of Property 2 and Property 3 of \Cref{lem:adaptive_sampling}.

The history $h=(T_1,T_2\dots,T_{\ell-1})$ is known since we
are conditioning on $H_i(h)$, so it follows that
\[
  f\parens*{S^*} - f\parens*{S_{i-1}}
  =
  \ExpCond{}{\OPT - \hat{f}\parens*{S_{i-1}}}{H_i(h)},
\]
because there is no randomness in the expectation.
Recall that $f(S_{i-1}) = \E[\hat{f}(S_{i-1}) \mid H_i(h)]$ \change{since}
the set $S_{i-1} = T_{1} \cup \dots \cup T_{\ell-1}$ is a union of complete sets, and hence
there are no partial, averaged contributions.
Rearranging the previous inequalities gives
\begin{align*}
  &\ExpCond{}{\hat{\Delta}\parens*{X_i, S_{i-1}}}{H_i(h)}
  \ge \frac{\parens*{1-\varepsilon}^2}{k}\cdot
      \ExpCond{}{\OPT - \hat{f}\parens*{S_{i-1}}}{H_i(h)},
\end{align*}
as desired.

Now we consider the case when $i > i^*$.
Because we condition on the history
$h=(T_1,T_2,\dots,T_{\ell-1})$
immediately before drawing a set $T_\ell$
that necessarily contains the $i$-th element,
the averaging property of~$\hat{f}$ and
the analysis for the previous case give us
\begin{align*}
  \ExpCond{}{\hat{\Delta}\parens*{X_i, S_{i-1}}}{H_i(h)}
  &= \ExpCond{}{ \hat{\Delta}\parens*{X_{i^*}, S_{i^* - 1}} }{H_i(h)}\\
  &\hspace{-2.0cm}\ge \frac{\parens*{1-\varepsilon}^2}{k}\cdot
      \ExpCond{}{\OPT - \hat{f}\parens*{S_{i^* - 1}}}{H_i(h)}\\
  &\hspace{-2.0cm}\ge \frac{\parens*{1-\varepsilon}^2}{k}\cdot
      \ExpCond{}{\OPT - \hat{f}\parens*{S_{i-1}}}{H_i(h)}.
\end{align*}
The final inequality makes use of 
\[
  \E[\hat{f}(S_{i-1}) \mid H_i(h)] \ge \E[\hat{f}(S_{i^* - 1}) \mid H_i(h)],
\]
which is a consequence of monotonicity and
the averaging property of $\hat{f}$.
This completes the proof for all $i \in [k]$.
\end{proof}

%\begin{lemma}
%\label{lem:prob_tree_flow}
%If for all $i \in [k]$, events $H_i(h)$ and any constant $c > 0$, we have
%\[
%  \ExpCond{}{\delta_i}{H_i(h)}
%  \le c\cdot
%  \ExpCond{}{\delta_{i-1}}{H_i(h)},
%\]
%then
%\[
%  \Exp{}{\delta_i}
%  \le c\cdot
%  \Exp{}{\delta_{i-1}}.
%\]
%\end{lemma}
%
%\begin{proof}
%\todo{This missing proof is critical but seems true via probability tree.}
%\end{proof}

\begin{lemma}
\label{lem:expectation_approx}
\change{Conditioned on event $Z$,}
if for all $i \in [k]$ and events $H_i(h)$ we have
\begin{align*}
  &\ExpCond{}{\hat{\Delta}\parens*{X_i, S_{i-1}}}{H_{i}(h)}
  \ge
  \frac{(1-\varepsilon)^2}{k}\cdot\ExpCond{}{\OPT - \hat{f}(S_{i-1})}{H_{i}(h)},
\end{align*}
then the algorithm returns a set $S_k$ such that
$\Exp{}{f\parens*{S_k}} \ge \parens*{1 - 1/e - \varepsilon} \OPT$.
\end{lemma}

\begin{proof}
Let $\delta_i = \OPT - \hat{f}(S_i)$, and observe that 
\[
  \ExpCond{}{\change{\hat\Delta}\parens*{X_i,S_{i-1}}}{H_i(h)}
  = 
  \ExpCond{}{\delta_{i-1}}{H_i(h)} - \ExpCond{}{\delta_{i}}{H_i(h)}
\]
by the linearity of expectation.
It follows from the assumption that
\[
  \ExpCond{}{\delta_{i}}{H_i(h)} \le \parens*{1 - \frac{(1-\varepsilon)^2}{k}}
  \cdot \ExpCond{}{\delta_{i-1}}{H_i(h)}.
\]
Since $\mathcal{H}_i$ conditioned on the event $H_i(h)$
is a partition of the final outcome distribution $\mathcal{H}$, 
it follows from the law of total probability that
\begin{align*}
  \Exp{}{\delta_{i}} &= \sum_{h \in \text{supp}\parens*{\mathcal{H}_i}}
    \ExpCond{}{\delta_{i}}{H_i(h)} \cdot \Prob{}{H_i(h)}\\
    &\le \parens*{1 - \frac{(1-\varepsilon)^2}{k}} \hspace{-0.2cm}\sum_{h \in \text{supp}\parens*{\mathcal{H}_i}}
      \hspace{-0.4cm}\ExpCond{}{\delta_{i-1}}{H_i(h)} \cdot \Prob{}{H_i(h)}\\
    &\le \parens*{1 - \frac{(1-\varepsilon)^2}{k}} \cdot \Exp{}{\delta_{i-1}}.
\end{align*}
Iterating this inequality over the sequence of expectations $\E[\delta_i]$ and
the using the fact $1-x \le e^{-x}$,
\begin{align*}
  \Exp{}{\delta_k} &\le \parens*{1 - \frac{(1-\varepsilon)^2}{k}}^k \cdot \Exp{}{\delta_{0}}
  \le \parens*{\frac{1}{e} + \varepsilon}\cdot\Exp{}{\delta_{0}}.
\end{align*}
We have
$\E[\hat{f}(S_0)] = \E[f(S_0)]$ and
$\E[\hat{f}(S_k)] = \E[f(S_k)]$
by the construction of $\hat{f}$.
\change{Further},
since~$f$ is nonnegative, we have $\delta_0 = \OPT - \hat{f}(S_0) \le \OPT$.
\change{Therefore, we have}
%\begin{align*}
$
  \Exp{}{f\parens*{S_k}} \ge \parens*{1 - 1/e - \varepsilon} \OPT,
$
%\end{align*}
which completes the proof.
\end{proof}

%\approx*

\begin{proof}[Proof of Theorem~\ref{thm:approx_factor}]
The cardinality constraint is satisfied by construction, so we start by
proving the adaptivity complexity of \Wrapper.
Lowering bounding $\OPT$ by $\Delta^*$ takes one adaptive round,
and each execution of the block in the parallelized for loop is independent of
all previous iterations. Therefore, it suffices to bound the adaptivity
complexity of the for loop block (\change{Lines 5--13}).
Each invocation of \AdaptiveSampling is potentially dependent on the last
since~$S$ can be updated in each round.
Therefore, because there are $m = O(1/\varepsilon)$ iterations in
the block, the total adaptivity complexity is
$O\parens{m \log(n/\delta) / \varepsilon}
= O\parens{\log(n/\delta)/\varepsilon^{2}}$
by \Cref{lem:adaptive_sampling}.

Now we analyze the query complexity of the algorithm.
Each call to \AdaptiveSampling behaves as intended with
probability at least $1-\hat{\delta}$, so by our choice of $\hat{\delta}$
and a union bound, all calls to \AdaptiveSampling are correct
with probability at least $1-\delta$
\change{(i.e., event $Z$ holds)}.
\change{If this is the case, then}
the expected query complexity of \AdaptiveSampling is~$O(n/\varepsilon)$
\change{by \Cref{lem:adaptive_sampling}}.
Therefore, it follows that the \change{overall} expected query complexity of \Wrapper is
$O(n + rm (n/\varepsilon) + \change{\delta n^2}) = O(n \log(k)/\varepsilon^{3} + \change{\delta n^2})$.

To prove the approximation guarantee, first observe that
$\Delta^* \le \OPT \le k \Delta^*$ by submodularity.
Therefore, we know that $\Delta^*/k \le \tau^* \le \Delta^*$.
The values of $\tau$ considered are $(1+\varepsilon)^i\Delta^*/k$, so by
our choice for the number of iterations $r = O(\log(\change{k})/\varepsilon)$,
there exists a $\tau$ satisfying
$\tau \le \tau^* \le (1+\varepsilon)\tau$.
Although we do not know this value of $\tau$,
we \change{can} use its existence to
give a guarantee by taking the maximum over all potential solutions.
Therefore, \change{conditioned on $Z$ which happens} with probability at least $1 - \delta$, we have
\[
  \Exp{}{f(S) \mid \change{Z}} \ge \parens*{1-1/e-\varepsilon}\OPT
\]
by \Cref{lem:conditional_iterative_property} and \Cref{lem:expectation_approx},
assuming the returned set satisfies $|S|=k$.
If instead $|S| < k$, then all unchosen elements $x \in N \setminus S$
have marginals $\Delta(x, S) \le \tau/4$ by our choice of $m$
and Property~3 of \Cref{lem:adaptive_sampling}.
Thus, for any $\tau \le \tau^*$, 
monotonicity and submodularity give
\begin{align*}
  f\parens*{S^*} &\le f(S) + \sum_{x \in S^*} \Delta(x, S)
    \le f(S) + k \tau / 4,
%    \le f(S) + \OPT / 4,
\end{align*}
which implies 
$f(S) \ge \parens{1 - 1/4} \OPT \ge \parens{1 - 1/e} \OPT$.
\change{Putting everything together and using the nonnegativity of $f$, we have
$\E[f(S)] \ge (1 - 1/e - \varepsilon)(1-\delta)\OPT$, which completes the proof.
}
\end{proof}

\section{Achieving Linear Query Complexity via Preprocessing}
\label{sec:linear-queries}

In this section we demonstrate different ways of using \AdaptiveSampling to
preprocess the interval containing $\OPT$ and
reduce the total query complexity of the algorithm without increasing its adaptivity.
In \Cref{sec:binary-search} we show
how to we reduce the ratio of the interval containing $\OPT$ from $O(k)$ to
$O(\log(k))$ in $O(\log\log(k))$ iterations of an imprecise binary search,
reducing the query complexity from $O(n \log(k))$ to $O(n \log\log(k))$.
In \Cref{sec:subsampling} we show how to reduce the ratio
of the interval from~$R$ to \change{$O(\log^5(R))$} \change{in each step,}
until \change{the ratio} is \change{a} constant.
\change{This second approach subsamples} the ground set
and uses the \change{imprecise} binary search subroutine.
By \change{adaptively setting} the parameters at each step according to the current ratio,
we reduce the query complexity to~$O(n)$ while maintaining $O(\log(n))$ adaptivity.

\subsection{Reducing the Query Complexity with an Imprecise Binary Search} 
\label{sec:binary-search}

%\todo{Discuss trade-off in a practical sense?
%This trade-off between adaptivity and query complexity
%is beneficial. For example, if we generously take $\varepsilon = 0.25$,
%then $\log\log(k) = \varepsilon^{-2.5}$
%for $k = e^{e^{32}}$, so in all practical settings, the complexity requirements
%imposed by $\varepsilon$ dominate $O(\log\log(k))$.
%\begin{itemize}
%  \item Adaptivity: $O(1) \mapsto O(\log\log(k))$
%  \item Queries: $O(\log(k)) \mapsto O(\log\log(k))$
%\end{itemize}
%}

To see how we can use a binary search,
consider the output of $\AdaptiveSampling(f, k, \tau, 1-p, \delta)$
for an arbitrary value of $\tau$.
If $|S| = k$, then by
Property~2 of \Cref{lem:adaptive_sampling} we have
$pk\tau \le \E[f(S)] \le \OPT$.
Otherwise, if $|S| < k$ then 
by Property~3 of \Cref{lem:adaptive_sampling}
we have $\Delta(x,S) \le~\tau$.
In the second case, it follows for monotone submodular functions that
$f(S) \le \OPT \le f(S) + k\tau$.
If $f(S) \le k\tau$ then $\OPT \le 2k\tau$,
and if $k\tau < f(S) \le \OPT$ then $pk\tau \le \OPT$.
Therefore, after each call to \AdaptiveSampling we can determine
with probability at least $1 - \hat\delta$ that
one of the following inequalities is true:
$\OPT \le 2k\tau$ or $pk\tau \le \OPT$.
Note that these decisions may overlap, hence the term
\change{\emph{imprecise binary search}}.
We give the guarantees \change{for} \BinarySearch below
and defer the proof of \Cref{cor:binary_search} to
\Cref{app:binary_search_proof}.

\begin{algorithm}[H]
  \caption{\BinarySearch}
  \label{alg:binary-search}
  \vspace{0.1cm}
  \textbf{Input:} evaluation oracle for $f : 2^N \rightarrow \change{\R_{\ge 0}}$, constraint $k$, error $\varepsilon$, failure probability $\delta$
  \begin{algorithmic}[1]
    \State Set maximum marginal $\Delta^* \leftarrow \max\{f(x) : x \in N\}$
    \State Set interval bounds $L \leftarrow \Delta^*$, $U\leftarrow k\Delta^*$
    \State Set balancing parameter $p \leftarrow 1/\log(k)$
    \State Set upper bound $m \leftarrow \change{\ceil{\log_2(\log(k))}}$
    \State Set smaller failure probability $\hat{\delta} \leftarrow \delta/(m+1)$
    \For{$i=1$ to $m$} 
      \State Set $\tau \leftarrow \sqrt{LU/(2p)}/k$
      \State Set $S \leftarrow \AdaptiveSampling(f, k, \tau, 1-p, \hat{\delta})$
      \If{$|S| < k$ and $f(S) \le k\tau$} \Comment{Imprecise binary search decision}
        \State Update $U \leftarrow 2k\tau$
      \Else
        \State Update $L \leftarrow pk\tau$
      \EndIf
    \EndFor
    \State \textbf{return} $\Wrapper(f,k,\varepsilon,\hat{\delta})$
                            modified to search over $[L/k,U/k]$
  \end{algorithmic}
\end{algorithm} 

\begin{restatable}[]{corollary}{binarySearch}\label{cor:binary_search}
For any monotone, nonnegative submodular function $f$,
the algorithm \BinarySearch outputs a set $S \subseteq N$ with $|S| \le k$ in
$O(\log(n/\delta)/\varepsilon^2)$ adaptive rounds
\change{and with $O(n \log\log(k) /\varepsilon^{3} + \delta n^2)$}
expected oracle queries \change{such that}
$\E[f(S)] \ge (1 - 1/e - \varepsilon)\change{(1-\delta)}\OPT$.
\end{restatable}

\subsection{Reducing the Query Complexity by Subsampling} 
\label{sec:subsampling}

\change{Now} we describe how to combine \AdaptiveSampling and subsampling
to preprocess the interval containing $\OPT$ in $O(\log(n))$ adaptive rounds
and with a total of $O(n)$ queries so that the final interval has a constant
ratio.
There are three main ideas underlying the algorithm 
\SubsamplePreprocess\change{:
\begin{enumerate}
    \item We subsample the ground set $N$ so that the query complexity of each
       \AdaptiveSampling call is sublinear (\Cref{lem:adaptive_sampling}).
    \item We relate the optimal solution in the sampled space to $\OPT$ via the subsampling probability.
    \item We repeatedly subsample the ground set $N$ with a granularity that depends on the current ratio of the feasible interval $R$.
    In each of these iterations, we run
    $O(\log(R))$ imprecise binary search decisions in parallel (by calling
\AdaptiveSampling with error $1-p$ as described in \Cref{sec:binary-search}) to
reduce the ratio from $R$ to $O(\log^5(R))$.
\end{enumerate}
}
\noindent
The adaptivity of each step is $O(\log(n/\change{\delta})/\log(1/p))$
by \Cref{lem:adaptive_sampling} because the calls are distributed.
There are $O(\log^*(R))$ ratio reduction rounds,\footnote{\change{We
note that $\log^*(n)$ denotes the \emph{iterated logarithm}
defined as $\log^*(n) := 1 + \log^*(\log(n))$ if $n > 1$ and $0$ otherwise.}}
but by our choice of \change{parameters} $\ell$ and $p$ in each round,
the total number of adaptive rounds is $O(\log(n/\delta))$.
Therefore, when \SubsamplePreprocess terminates,
the interval containing $\OPT$ has a constant ratio.
\change{In the last step,} we run \Wrapper modified to search over this new interval
for the final solution.

Now we formally present \SubsamplePreprocess and state the lemmas that are
prerequisites for its guarantees.  All proofs \change{for this} preprocessing \change{step}
are deferred to \Cref{sec:subsample_proofs}.
We first show how the optimal solution
in a subsampled ground set \change{is related} to $\OPT$
in terms of the subsampling probability.

\begin{restatable}[]{lemma}{subsampleApprox}
\label{lem:subsample_approx}
For any \change{monotone, nonnegative} submodular function $\change{f}$,
sample each element in $N$ independently with probability $1/\ell$.
\change{Let the subsampled} set be $N'$, 
and denote the optimal solution \change{restricted to} $N'$ by $\OPT'$.
Let $\Delta^*$ be an upper bound for the max marginal \change{gain} in $N$.
\change{Then, for any $\delta \in (0, 1]$,}
with probability at least $1 - \delta$, \change{we have}
\[
  \frac{1}{2}\parens*{\Delta^* + \OPT'}
  \le \OPT \le
  \frac{2\ell}{\delta}\parens*{\Delta^* + \OPT'}.
\]
\end{restatable}

\begin{algorithm}
  \caption{\SubsamplePreprocess}
  \label{alg:subsample-preprocess}
  \vspace{0.1cm}
  \textbf{Input:} evaluation oracle for $f : 2^N \rightarrow \change{\R_{\ge 0}}$, constraint $k$, error $\varepsilon$, constant failure probability $\delta$
  \begin{algorithmic}[1]
    \State Set maximum marginal $\Delta^* \leftarrow \max\{f(x) : x \in N\}$
    \State Set interval bounds $L \leftarrow \Delta^*$, $U\leftarrow k\Delta^*$,
      $R^* \leftarrow \morteza{2000}$
    \While{$U/L \ge R^*$} \label{algline:while-loop}
      \State \change{Set $R \gets U / L$}
      \State Set sampling ratio $\ell \leftarrow \log^2(R)$
      \State Set imprecise decision accuracy $p \leftarrow 1/\log(R)$
      \State Set upper bound $m = \ceil{\log_2(R)}$
      \State Set smaller failure \change{probability} $\hat\delta_R \leftarrow \delta/(2(m+1) \log(R))$
      \State Set $N' \leftarrow \textsc{Subsample}(N, 1/\ell)$
      \For{$i = 0$ to $m$ in parallel}
        \State Set $\tau_i \leftarrow 2^i(L/k)$
        \State Set $S_i' \leftarrow \AdaptiveSampling(f', k, \tau_i,$ $1 - p, \hat\delta_R \change{/n})$
        \State Decide if \change{$\OPT' \le 2k\tau_i$
               or $\OPT' \ge pk\tau_i$}
               using $S_i'$ \Comment{Imprecise binary search decision}
      \EndFor
      \If{\change{$\OPT' \le 2k\tau_0$}} \Comment{Update interval}
        \State Update $L \leftarrow (\Delta^* + L)/2$ 
        \State Update $U \leftarrow (4\ell/\hat\delta_R)(\Delta^* + L)$
      \ElsIf{\change{$\OPT' \ge pk\tau_m$}}
        \State Update $L \leftarrow \change{pU}$ %(p/2)(\Delta^* + U)$
        \State Update $U \leftarrow \change{U}$ \change{\Comment{$U$ stays intact}} % (2\ell/\hat\delta_R)(\Delta^* + U)$
      \Else
        \State Set $i^* \leftarrow$ first $i=0$ to $m$
              such that \change{ $\OPT' \ge pk \tau_i$ and $\OPT' \le 2k \tau_{i+1}$}
        \State Update $L \leftarrow (p/2)(\Delta^* +  \change{2^{i^*} L})$
        \State Update $U \leftarrow (\change{8}\ell/\hat\delta_R)(\Delta^* +  \change{2^{i^*} L})$
      \EndIf
      \morteza{
      \If{$\log(\log(R)) < 2 \log(\log(U/L))$} \label{algline:if-condition-loglog-decline} \Comment{$R$ is sufficiently small}
        \State \textbf{break} \label{algline:loglog-did-not-decrease-enough}
      \EndIf }
    \EndWhile
    \State \textbf{return} $[L, U]$
  \end{algorithmic}
\end{algorithm} 

Next we show that in each round of \SubsamplePreprocess, the current ratio $R$
becomes polylogarithmically smaller until it drops below a constant lower bound
threshold $R^*$.
The adaptivity and query complexity of this iteration is sublinear
\change{in} $R$, so by summing over the $O(\log^*(k))$
rounds of \SubsamplePreprocess, the
total number of adaptive rounds and expected number of queries are
$O(\log(n/\delta))$ and~$O(n)$, respectively.

\begin{restatable}[]{lemma}{reduceRatio}
\label{lem:reduce_ratio}
For any monotone, \change{nonnegative} submodular function $f$,
let $[L,U]$ be an interval containing $\OPT$ with $U/L = R$.
For any ratio $R > 0$, with probability at least $1 - \delta$,
we can compute a new feasible interval with ratio \change{at most $(64/\delta) \log^5(R)$} 
such that:
\begin{itemize}
  \item The number of adaptive rounds is
    $O\parens*{\frac{\log\parens*{n/\delta}}{\log\log(R)}}$.
  \item The \change{expected} number of queries is $O\parens{n/\log(R)}$.
\end{itemize}
\end{restatable}

\begin{restatable}[]{lemma}{linearPreprocess}
\label{lem:linear_preprocess}
For any monotone, \change{nonnegative} submodular function $f$ and constant $0 < \delta \le 1$,
with probability at least $1 - \delta$,
the algorithm \SubsamplePreprocess
returns an interval containing $\OPT$ with ratio 
\morteza{$O(\alpha^{4\log(\alpha)})$} 
%$O(1/\delta^2)$
in $O(\log(n/\delta))$ adaptive rounds
and uses $O(n)$ queries in expectation,
\morteza{where $\alpha = 64/\delta$}.
\end{restatable}

Last, we show how to use the reduced interval returned by
\SubsamplePreprocess with the \Wrapper~\change{algorithm} to get
\Subsample.

\begin{restatable}[]{theorem}{linearAlgorithm}
\label{lem:linear_algorithm}
For any monotone, nonnegative
submodular function $f$
and constant $0 < \varepsilon \le 1$,
the algorithm \Subsample outputs a set $S \subseteq N$ with $|S| \le k$ in
$O(\log(n\change{/\varepsilon})/\varepsilon^2)$ adaptive rounds
\change{and with $O(n\log\morteza{^2}(1/\varepsilon) / \varepsilon^{3})$ expected queries}
such that
$\E[f(S)] \ge (1 - 1/e - \varepsilon)\OPT$.
\end{restatable}

\begin{proof}
Set a smaller error $\hat\varepsilon = \varepsilon/4$ and run
$\SubsamplePreprocess(f,k,\hat\varepsilon,\change{\hat\varepsilon})$ to obtain an interval with
ratio~\morteza{$O\parens{\alpha^{4\log(\alpha)}}$}
that contains $\OPT$ with probability at least
\morteza{$1 - \hat\varepsilon$ where $\alpha = 64 / \hat\varepsilon$}.
Next, modify and run $\Wrapper(f, k, \hat\varepsilon, \change{\hat\varepsilon / n})$
so that it searches over the interval with ratio~\morteza{$O(\alpha^{4\log(\alpha)})$}.
\morteza{Searching this range requires $O\parens{\log_{1+\hat\varepsilon} \parens*{\alpha^{4\log(\alpha)}}} = O\parens{\log^2(1/\varepsilon)/\varepsilon}$ iterations.} %$O(\varepsilon^{-2})$.
Both $\SubsamplePreprocess$ and $\Wrapper$ succeed with probability at least
$1 - 2\hat\varepsilon$ by a union bound.
Therefore, conditioning on the success of both events, we have
\begin{align*}
  \Exp{}{f(S)} &\ge \parens*{1 - 1/e - \hat\varepsilon}\OPT \cdot
    \parens*{1 - 2\hat\varepsilon}
  \ge (1 - 1/e - \varepsilon)\OPT,
\end{align*}
as desired.

The \change{adaptivity} complexity follows from the
guarantees of \Cref{lem:linear_preprocess} and Theorem~\ref{thm:approx_factor}.
\change{For the query complexity, observe that if the preprocessed interval does
not have ratio $O(\alpha^{4 \log(\alpha)})$,
then we can just output the empty set without calling $\Wrapper$
(to avoid a superlinear number of queries),
as this happens with probability at most $\hat{\varepsilon}$ by \Cref{lem:linear_preprocess}.
}
\end{proof}

\section{Using the \AdaptiveSampling Algorithm for Submodular Cover}
\label{sec:submodular-cover}
In the submodular cover problem, we aim to find a minimum cardinality subset
$S$ such that $f(S)$ is at least some target goal $L$. In some sense,
this problem can be viewed as the dual of submodular maximization
with a cardinality constraint. To formalize the submodular
cover problem, we want to solve $\min_{S \subseteq N} |S|$ subject to the value
lower bound $f(S) \geq L$.  We overload the notation $S^*$ to denote the
lexicographically least minimum size set
satisfying the value lower bound. Therefore, the value of
$\OPT$ is the cardinality $|S^*|$.  To overcome granularity issues resulting from
arbitrarily small marginal gains, a standard assumption is to work with
\emph{integer-valued submodular functions}. %(i.e., $f : 2^N \rightarrow \Z_{\geq 0}$). 

The greedy algorithm provides the state-of-the-art approximation for submodular
cover by outputting a set of size $O(\log(L) |S^*|)$.
There have been recent attempts \change{(e.g.,~\cite{MZK-NIPS16})}
to develop distributed
algorithms based on the greedy approach
that achieve similar approximation factors,
but these algorithms have suboptimal adaptivity complexity because
the summarization algorithm of the centralized machine is sequential.
Here, we show how to apply the ideas behind the \AdaptiveSampling algorithm to
submodular cover so that the algorithm runs in a logarithmic number of adaptive
rounds without losing the approximation guarantee. 

We start by giving a high-level description of our algorithm.
Similar to \cite{MZK-NIPS16}, which attempts to imitate the greedy algorithm,
we initialize $S=\emptyset$ and set the threshold $\tau$ to the highest
marginal value $\change{\Delta^* = \max_{x \in N} f(x)}$.
Then \change{the algorithm} repeatedly adds sets of items to~$S$
whose average value to $S$ is at least $(1-\varepsilon)\tau$.
\change{When} we run
out of high value items, we lower the threshold from $\tau$ to
$(1-\varepsilon)\tau$ and repeat the process.
Unlike the cardinality constraint setting,
the stopping condition of this algorithm
is when the value of $f(S)$ reaches the \change{target value} $L$. 

Specifically, for each threshold $\tau$ we run a variant of \AdaptiveSampling
called \AdaptiveSamplingForCover~\change{(\Cref{alg:sampling-for-cover})}
as a subroutine to find a maximal set of 
valuable items in $O(\log(n))$ adaptive rounds.
The first difference between this algorithm and \AdaptiveSampling is that it
takes the value lower bound $L$ as part of its input instead of a cardinality
constraint~$k$. 
Therefore, we slightly modify the \AdaptiveSampling algorithm as follows.
Since we do not have an explicit constraint on the number of elements
that we can to add,
we set $m$ such that it is possible to add all elements at once.
This change is reflected in Lines~$2$~and~$13$ of \AdaptiveSamplingForCover.
\change{Next,} we use
\[
  \change{k = \ceil*{\frac{L-f(S)}{(1-\varepsilon)\tau}}}
\]
as an upper bound for the number of elements that can be added in each \change{round}.
This is a consequence of our initial choice of $\tau = \Delta^*$
and the method for lowering the threshold as the algorithm progresses.
We \change{know} that for the current \change{threshold} $\tau$,
no element has marginal gain more than $\tau/(1-\change{\varepsilon})$
to the selected set~$S$
\change{by Property 3 of \Cref{lem:adaptive_sampling}.} 
Furthermore, the average contribution of elements in $T$
\change{for this stage}
satisfies $\E[\Delta(T,S)/|T|] \ge (1-\change{\varepsilon})\tau$
\change{by an analog of Property 2 of \Cref{lem:adaptive_sampling}.} 
\change{This then justifies} our choice for the upper bound.

To give intuition for why this leads to an acceptable approximation factor,
observe that for the current threshold $\tau$, the optimum
(conditioned on our current choice of $S$) must have at least
$(L - f(S))/(\tau/(1-\change{\varepsilon}))$ elements since the marginal gains are
\change{upper} bounded.
Our cardinality bound $\change{\ceil{(L-f(S))/((1-\varepsilon)\tau)}}$
implies that the algorithm does not add too many more elements than the
optimum in each round.
The final modification is \change{the} stopping condition on Line~$15$,
where we check whether or not we have reached the value lower bound $L$.

\begin{algorithm}[H]
  \caption{\AdaptiveSamplingForCover}
  \label{alg:sampling-for-cover}
  \vspace{0.1cm}
  \textbf{Input:} evaluation oracle for $f : 2^N \rightarrow \change{\mathbb{Z}_{\ge 0}}$, value goal $L$, threshold $\tau$,
    error $\varepsilon$, failure probability $\delta$
  \begin{algorithmic}[1]
    \State Set smaller error $\hat{\varepsilon} \leftarrow \varepsilon/3$
    \State Set iteration bounds
        $r \leftarrow \ceil{\log_{(1-\hat\varepsilon)^{-1}}(2n/\delta)}$,
        $m \leftarrow \ceil{\log_{\change{(1+\hat\varepsilon)}}(n)}$
    \State Set smaller failure probability
        $\hat{\delta} \leftarrow \delta/(2r(m+1))$
    \State Initialize $S \leftarrow \emptyset$, $A \leftarrow N$
    \For{$r$ rounds}
      \State Filter $A \leftarrow \{x \in A : \Delta(x,S) \ge \tau\}$
      \If{$\abs{A} = 0$}
        \State \textbf{break}
      \EndIf
      \For{$i=0$ to $m$} 
        \State Set $t \leftarrow \min\{\floor{(1 + \hat{\varepsilon})^i}, |A|\}$
        \If{$\EstimateMean(\mathcal{D}_t, \hat{\varepsilon},
          \hat{\delta} )$} %\Comment{\Cref{def:indicator_distribution}}
          \State \textbf{break}
        \EndIf
      \EndFor
      \State Set $T \sim \mathcal{U}\parens*{A, 
                \min\set*{t,\change{\ceil*{(L-f(S))/((1-\varepsilon)\tau)}}}}$
      \State Update $S \leftarrow S \cup T$
      \If{$f(S) \geq L$}
        \State \textbf{break}
      \EndIf
    \EndFor
    \State \textbf{return} $S$
  \end{algorithmic}
\end{algorithm} 

\begin{corollary}
\label{cor:threshold-cover}
\change{Let $Z$ be the event that all calls to \EstimateMean give correct outputs
(i.e., the reported property in \Cref{lem:estimator} holds).
For any integer-valued, monotone, nonnegative submodular function $f$,}
\AdaptiveSamplingForCover outputs $S \subseteq N$
in $O(\log(n/\delta)/\varepsilon)$ adaptive rounds
such that %$f(S) \le L/(1-\varepsilon)$. Furthermore,
the following properties hold \change{conditioned on $Z$}:
\begin{enumerate}
%  \item There are $O(n/\varepsilon)$ oracle queries in expectation.
  \item The \change{average marginal gain satisfies} $\E[f(S)/|S|] \ge (1-\varepsilon)\tau$.
  \item \change{With probability at least $1 - \delta/2$,}
    if $f(S) < L$, then $\Delta(x,S) < \tau$ for all $x \in N$.
\end{enumerate}
\change{Further, event $Z$ happens with probability at least $1 - \delta/2$.}
\end{corollary}

\begin{proof}
The proof is a direct consequence of the proof for
\Cref{lem:adaptive_sampling}.
\end{proof}

As explained above, we iteratively use \AdaptiveSamplingForCover as a
subroutine starting from the highest threshold $\tau=\Delta^*$ to ensure that
we either reach the value goal $L$ or that there is no element with marginal
value above~$\tau$ to the current set $S$.  If we have not reached the value lower
bound $L$, we reduce the threshold by a factor of $(1-\varepsilon)$ and
repeat.  This idea is summarized in the \WrapperForCover algorithm.
We note that by the integrality assumption on $f$, this algorithm
is guaranteed to output a feasible solution in a deterministic amount of
time, since the threshold can be lowered enough such that
any element \change{with positive marginal gain} can \change{eventually}
be added to the solution.

\begin{algorithm}[H]
  \caption{\WrapperForCover}
  \label{alg:wrapper-for-cover}
  \vspace{0.1cm}
  \textbf{Input:} evaluation oracle for $f : 2^N \rightarrow \change{\mathbb{Z}_{\ge 0}}$,
  value goal $L$ %, error $\varepsilon$
  \begin{algorithmic}[1]
%    \State \todo{Set smaller error: $\hat{\varepsilon} \leftarrow \varepsilon/2$? Propagate change.}
    \State Set error $\varepsilon \leftarrow 1/2$
    \State Set upper bounds $\Delta^* \leftarrow \change{\max\{f(x) : x \in N\}}$,
    $m \leftarrow \ceil{\log(\Delta^*)/\varepsilon} + \change{1}$
    \State Set failure probability $\delta \leftarrow 1/(n(m+1))$
      \State Initialize $S \leftarrow \emptyset$
      \For{$i=0$ to $m$} \Comment{Until $(1-\varepsilon)^i\Delta^* < 1$}
        \State Set $\tau \leftarrow (1-\varepsilon)^i \Delta^*$
        \State Set $T \leftarrow \AdaptiveSamplingForCover(N, f_S, L - f(S), \tau, \varepsilon, \delta)$
        \State Update $S \leftarrow S \cup T$
        \If{$f(S) \geq L$}
          \State \textbf{break}
        \EndIf
      \EndFor
    \State \textbf{return} $S$
  \end{algorithmic}
\end{algorithm} 

\begin{restatable}[]{theorem}{submodularCover}
\label{thm:submodular_cover}
For any integer-valued, nonnegative, monotone submodular function $f$,
the algorithm \WrapperForCover outputs a subset $S \subseteq N$ 
with $f(S) \ge L$ in
$O(\log(n\log(L))\log(L))$
adaptive rounds such that %with probability at least $1 - \delta$
$\E[|S|] = O(\log(L)|S^*|)$.
\end{restatable}

\noindent
We defer the proof of Theorem~\ref{thm:submodular_cover} to
\Cref{app:submodular-cover} and remark that it follows a similar line of
reasoning to the analysis of the approximation factor for \Wrapper
in \Cref{lem:expectation_approx}.
The main difference between these proofs is that for
Theorem~\ref{thm:submodular_cover} we need to show that \WrapperForCover
makes geometric progress towards its value constraint~$L$
not only in expectation, but also with constant probability.
We do this by considering the progress of the algorithm in intervals
of $O(|S^*|)$ elements, which ultimately
allows us to analyze $\E[|S|]$
using \change{properties of} the negative binomial distribution.
Lastly, we have not optimized the constant in the approximation
factor, but one could do this by more carefully considering
how the error $\varepsilon$
affects the lower bound for the constant probability term in our analysis.

\bibliographystyle{alpha}
\bibliography{references}

\appendix

\newpage
\section{Missing Analysis from \Cref{sec:threshold}}
\label{app:reduced-mean}

\subsection{Proof of \Cref{lem:percentage_nonincreasing}}
\label{app:nonincreasing_proof}

\percentageNonincreasing*

\change{
\begin{proof}
We have $\E[I_0] = 1$ by the definition of $I_t$ and the fact that the candidates
in $A$ are filtered at the beginning of each round.

Now, let $m=|A|$ and
$(n)_t = \prod_{i=1}^t (n-(i-1))$ denote the falling factorial.
For $t \in \{1,\dots,|A|-1\}$,
summing over all $(t+1)$-truncated permutations of the elements $x_1,x_2,\dots,x_m \in A$
gives
\begin{align*}
    \Exp{}{I_{t}}
    &= \frac{1}{(m)_{t+1}} \sum_{x_1,\dots,x_{t},x_{t+1}}
      \ind\bracks*{\Delta\parens*{x_{t+1}, S\cup\set{x_1,\dots,x_t}} \ge \tau} \\
    &\le \frac{1}{(m)_{t+1}} \sum_{x_1,\dots,x_{t},x_{t+1}}
      \ind\bracks*{\Delta\parens*{x_{t+1}, S\cup\set{x_1,\dots,x_{t-1}}} \ge \tau}
      & \hfill \text{(submodularity)} \\
    &= \frac{m-t}{(m)_{t+1}} \sum_{x_1,\dots,x_{t-1},x_{t+1}}
      \ind\bracks*{\Delta\parens*{x_{t+1}, S\cup\set{x_1,\dots,x_{t-1}}} \ge \tau} \\
    &= \frac{1}{(m)_{t}} \sum_{x_1,\dots,x_{t-1},x_{t}}
      \ind\bracks*{\Delta\parens*{x_{t}, S\cup\set{x_1,\dots,x_{t-1}}} \ge \tau}
      & \hfill \text{(symmetry)} \\
    &=
    \Exp{}{I_{t-1}}.
\end{align*}
The second-to-last equality is a change of variables. Last,
the boundary case $\E[I_{|A|-1}] \ge \E[I_{|A|}] = 0$ holds
by the definition of $I_{t}$.
\end{proof}
}

\subsection{Analysis of \EstimateMean Algorithm}
\label{app:estimator_analysis}

\begin{lemma}[Chernoff bounds, \cite{bansal2006santa}]
\label{lem:chernoff}
Suppose $X_1,\dots,X_n$ are binary random variables such that
  $\Prob{}{X_{i}=1} = p_i$. Let $\mu = \sum_{i=1}^n p_i$ and
$X = \sum_{i=1}^n X_i$. Then for any $a > 0$, we have
\[
  %\Prob{}{X - \mu \ge a} \le e^{a - (a+\mu)\log(1+\frac{a}{\mu})}.
  \Prob{}{X - \mu \ge a} \le e^{-a \min\parens*{\frac{1}{5}, \frac{a}{4\mu}}}.
\]
Moreover, for any $a > 0$, we have
\[
  \Prob{}{X - \mu \le - a} \le e^{-\frac{a^2}{2\mu}}.
\]
\end{lemma}

\estimator*

\begin{proof}
By construction the number of samples used is
$m = 16 \ceil{\log(2/\delta)/\varepsilon^2}$.
To show the correctness of \EstimateMean,
it suffices to prove that
$\Prob{}{\abs*{\overline{\mu} - \mu} \ge \varepsilon/2} \le \delta$.
Letting $X = \sum_{i=1}^m X_i$, this is equivalent to
\[
  \Prob{}{\abs*{X - m\mu} \ge \frac{\varepsilon m}{2}}
  \le \delta.
\]

Using the Chernoff bounds in \Cref{lem:chernoff} and a union bound,
for any $a > 0$ we have
\[
  \Prob{}{\abs*{X - m\mu} \ge a} \le
  e^{-\frac{a^2}{2m\mu}} + e^{-a \min\parens*{\frac{1}{5}, \frac{a}{4m\mu}}}.
\]
Let $a = \varepsilon m /2$ and consider the exponents of the two terms 
separately.
Since $\mu \le 1$, we bound the left term by
\[
  \frac{a^2}{2m\mu} = \frac{\varepsilon^2 m^2}{8m\mu}
  \ge \frac{\varepsilon^2}{8\mu} \cdot \frac{16 \log(2/\delta)}{\varepsilon^2}
  \ge \log(2/\delta).
\]
For the second term, first consider the case when $1/5 \le a/(4m\mu)$.
For any $\varepsilon \le 1$, it follows that
\[
  a \min\parens*{\frac{1}{5},\frac{a}{4m\mu}} = \frac{1}{5}
  \ge \frac{\varepsilon}{10} \cdot \frac{16\log(2/\delta)}{\varepsilon^2}
  \ge \log(2/\delta).
\]
Otherwise, we have $a/(4m\mu) \le 1/5$, and
by previous analysis we have
$a^2/(4m\mu) \ge \log(2\delta)$.
Therefore, in all cases we have
\[
  \Prob{}{\abs*{X - m\mu} \ge \frac{\varepsilon m}{2}} 
  \le 2e^{-\log(2/\delta)}\\
  = \delta,
\]
which completes the proof.
\end{proof}

\section{Missing Analysis from \Cref{sec:linear-queries}}
\label{app:maximization}

\subsection{Analysis of \BinarySearch Algorithm}
\label{app:binary_search_proof}

\binarySearch*

\begin{proof}
At the beginning of the algorithm, the interval $[L,U] = [\Delta^*, k\Delta^*]$
contains $\OPT$ by submodularity.
In each step of the binary search we can choose $\tau \in [L,U]$
and use \AdaptiveSampling to reduce the interval by some amount such that
the updated interval contains $\OPT$.
This decision process is described in \Cref{sec:binary-search}.
Our goal is to run \Wrapper on a smaller feasible interval with ratio
$U/L = O(1/p)$ so that we can set
$r = O(\log(1/p)/\varepsilon)$ instead of $O(\log(k)/\varepsilon)$.
This objective stems from the fact that \Wrapper grows $(1+\varepsilon)^i$-sized
balls until the interval is covered to approximate $\tau^*$.
Therefore, at each step of the binary search we let
\[
  \tau = \argmin_{\tau' \in [L,U]} \max \set*{\frac{2k\tau'}{L}, \frac{U}{pk\tau'}},
\]
by considering the worst ratio of both outcomes.
Since one function is increasing in $\tau$ and the other is decreasing,
we equate the two expressions to optimize~$\tau$,
which gives us $\tau = \sqrt{UL/(2pk^2)}$.
It follows that the ratio of the updated interval is
at most $\sqrt{2U/(pL)}$.

\change{
Starting with a ratio $R = U/L$,
it follows for any $p \in (0, 1]$, the $i$-th interval ratio is at most
\begin{align*}
    \parens*{\frac{2}{p}}^{\frac{1}{2} + \frac{1}{4} + \dots + \frac{1}{2^i}}
    \cdot
    R^{\frac{1}{2^i}}
    &=
    \parens*{\frac{2}{p}}^{1 - \frac{1}{2^{i+1}}}
    \cdot
    R^{\frac{1}{2^i}} \\
    &\le
    \frac{2}{p} \cdot R^{\frac{1}{2^i}}.
\end{align*}
Initially $R = k$, so
setting $m = \ceil{\log_{2}(\log(k))}$
gives
\begin{align*}
    \frac{2}{p} \cdot R^{\frac{1}{2^m}}
    &\le
    \frac{2}{p}
    \cdot
    R^{\frac{1}{2^{\log_2(\log(k))}}} \\
    &=
    \frac{2}{p} \cdot R^{\frac{1}{\log(k)}} \\
    &=
    \frac{2}{p} \cdot k^{\frac{1}{\log(k)}} \\
    &=
    \frac{2e}{p}.
\end{align*}
Setting $p = 1/\log(k)$, the final ratio is at most $2e \log(k)$.
}
% Note: Previous proof seems questionable...
% To track the progress of the binary search, it is convenient
% to analyze the logarithm of the ratio.
% Each step maps
% $\log\parens{U/L} \mapsto (1/2)\log\parens{U/L} + (1/2)\log(2/p)$,
% so after the $i$-th step the log of the ratio of the interval is at most
% \begin{align*}
%   &\parens*{1/2}^{i} \log\parens*{U/L} + \sum_{j=1}^i
%   \parens*{1/2}^{j} \log\parens*{2/p}
%   \parens*{1/2}^{i} \log\parens*{U/L} + \log\parens*{2/p}.
% \end{align*}
% Letting $m$ be the first step where $(1/2)^i \log(U/L) < 1$
% and recalling that $U/L = k$, it follows that
% \[
%   m = \floor*{\frac{\log\log(k)}{\log(2)}}.
% \]
% This means that after $m$ steps of the imprecise binary search, the ratio
% of the remaining interval is at most $2e/p$.
Running \Wrapper on \change{this} preprocessed \change{$[L,U]$} interval,
it suffices to set $r = \ceil{2\log(2e/p)/\varepsilon}$.% = O(\log(\log(k))/\varepsilon)$.
%Since each iteration of the binary search depends on the result of the last,
%we choose $p$ such that the preprocessing adaptivity is $O(\log(n/\delta))$.

The adaptivity complexity of \AdaptiveSampling is
$O(\log(n/\delta)/\log(1/p))$ by \Cref{lem:adaptive_sampling},
\change{so} it follows from the number of iterations $m$ in the binary search 
of \BinarySearch that
the adaptivity complexity of \change{the entire preprocessing step} is
\begin{align*}
  O\parens*{m \cdot \frac{\log(n/\delta)}{\log(1/p)}}
  = O\parens*{\log\log(k) \cdot \frac{\log(n/\delta)}{\log\log(k)}}
  = O\parens*{\log(n/\delta)}.
\end{align*}
Thus, the overall adaptivity of \BinarySearch is
$O(\log(n/\delta) / \varepsilon^{2})$
by Theorem~\ref{thm:approx_factor}.

Now we analyze the expected query complexity of \BinarySearch.
By our choice of $\hat\delta$ and a union bound,
assume all subroutines produce their guaranteed outputs
\change{(i.e., event $Z$ holds for all calls to \EstimateMean)}.
Each call to \AdaptiveSampling in the binary search makes
$O(n/(1-p))$ oracle queries in expectation by \Cref{lem:adaptive_sampling}.
Therefore, the total expected query complexity for the binary search \change{preprocessing} is
\begin{align*}
  O\parens*{m \cdot \frac{n}{1-p}}
  =
  O\parens*{\log\log(k) \cdot \frac{n}{1 - \frac{1}{\log(k)}}}
  =
  O\parens*{n \log\log(k)}.
\end{align*}
Next, since the ratio of the updated interval $[L,U]$ after the binary search is
$O(\log(k))$, it follows that by modifying the search for $\tau^*$
in \Wrapper, the expected query complexity in this stage is
$O\parens*{n \log\log(k) /\varepsilon^3}$
by Theorem~\ref{thm:approx_factor}. 
This term dominates the query complexity of the binary search,
so the result follows.
The approximation factor then holds
by Theorem~\ref{thm:approx_factor} \change{since} the updated interval contains $\tau^*$.
\end{proof}

\subsection{Analysis of \Subsample Algorithm}
\label{sec:subsample_proofs}

\begin{lemma}[Chebyshev's inequality]
\label{lem:chebyshev}
Let $X_1,X_2,\dots,X_{n}$ be independent random variables
with $\E[X_i] = \mu_i$ and $\Var(X_i) = \sigma_i^2$.
Then, for any $a > 0$,
\[
  \Prob{}{\abs*{\sum_{i=1}^n X_i - \sum_{i=1}^n \mu_i} \ge a}
  \le
  \frac{1}{a^2} \sum_{i=1}^n \sigma_i^2.
\]
\end{lemma}

\subsampleApprox*

\begin{proof}
Let $x_1,x_2,\dots,x_k$ be the elements in $S^*$ in lexicographic order.
By summing the marginal gain for each element when they are added in
lexicographic order, we have
\[
  f\parens*{S^*} = \sum_{x \in S^*} \Delta\parens*{x, \pi_x},
\]
where $\pi_x$ denotes the set of elements before $x$ in the
lexicographic order.
Subsample the ground set~$N$ such that each element is included
in the set $N'$ independently with probability $1/\ell$, and
let~$S'$ be the random set denoting the elements in $S^*$ that
remain after subsampling.
It follows from submodularity that for any value that $S'$ takes,
we have
\begin{align*}
  f\parens*{S'} \ge \sum_{x \in S'} \Delta(x, \pi_x).
\end{align*}
For each $x \in S^*$ define the random variable
\begin{align*}
  Z_{x} =
    \begin{cases}
      \Delta(x, \pi_{x}) & \text{with probability $1/\ell$},\\
      0 & \text{otherwise}.
    \end{cases}
\end{align*}
It follows that
\begin{align*}
  \Exp{}{Z_x} &= \Delta(x, \pi_x) \cdot \frac{1}{\ell}\\
  \Var\parens*{Z_x} &= \Delta(x,\pi_x)^2 \cdot \frac{1}{\ell} \parens*{1 - \frac{1}{\ell}}.
\end{align*}
Let $g(S^*)$ be the random variable
\[
  g(S^*) = \sum_{x \in S^*} Z_x,
\]
which is always a lower bound for the optimal solution \change{$\OPT'$} in $N'$.
It follows that
\[
  \Exp{}{g(S^*)} = \frac{1}{\ell} \sum_{x \in S^*} \Delta(x, \pi_x)
  = \frac{1}{\ell} \cdot f(S^*) = \frac{\OPT}{\ell}.
\]

%Let $\OPT'$ denote the optimal solution in $N'$.
Ultimately, we want to show that with probability at least $1 - \delta$,
we have
\begin{equation}
  \frac{\ell \Delta^*}{\delta} + \ell \cdot \OPT'
  \ge \OPT
  \ge \OPT',
\end{equation}
as this implies the lower and upper bounds
\[
  \frac{\ell}{\delta} \cdot \parens*{\Delta^* + \OPT'}
  \ge \OPT 
  \ge \frac{\OPT' + \Delta^*}{2}.
\]
Consider the probability
\begin{align*}
  \Prob{}{\frac{\ell \Delta^*}{\delta} + \frac{\OPT}{2} + \ell g(S^*) \ge \OPT}
  &=
  \Prob{}{\ell g(S^*) - \OPT \ge - \frac{\ell \Delta^*}{\delta} - \frac{\OPT}{2}}\\
  &=
  \Prob{}{\OPT - \ell g(S^*) \le \frac{\ell \Delta^*}{\delta} + \frac{\OPT}{2}}\\
  &= \Prob{}{\frac{\OPT}{\ell} - g(S^*) \le \frac{\Delta^*}{\delta} + \frac{\OPT}{2\ell}}.
\end{align*}
Using \change{Chebyshev’s inequality (\Cref{lem:chebyshev})}, 
the probability of the complementary event is
\begin{align*}
  \Prob{}{\frac{\OPT}{\ell} - g(S^*)  > \frac{\Delta^*}{\delta} + \frac{\OPT}{2\ell}}
  &\le \Prob{}{\frac{\OPT}{\ell} - g(S^*)  \ge \frac{\Delta^*}{\delta} + \frac{\OPT}{2\ell}}\\
  &\le \Prob{}{\abs*{\frac{\OPT}{\ell} - g(S^*)}  \ge \frac{\Delta^*}{\delta} + \frac{\OPT}{2\ell}}\\
  &\le \frac{1}{\parens*{ \Delta^*/\delta + \frac{\OPT}{2\ell}}^2}
  \cdot \frac{1}{\ell} \parens*{1 - \frac{1}{\ell}}
  \sum_{x \in S^*} \Delta(x, \pi_x)^2\\
  &\le \frac{1}{\parens*{ \Delta^*/\delta + \frac{\OPT}{2\ell}}^2}
  \cdot \parens*{\frac{\ell - 1}{\ell^2}}
  \Delta^* \cdot \sum_{x \in S^*} \Delta(x, \pi_x)\\
  &\le \frac{4 (\ell-1)}{\parens*{2\ell \Delta^*/\delta + \OPT}^2}
  \cdot \Delta^* \cdot \OPT\\
  %&\le \frac{4 \ell \Delta^* \OPT}{\parens*{2\ell \Delta^*/\delta + \OPT}^2}\\
  &\le \frac{4 \ell \Delta^* \OPT}{4\ell^2 (\Delta^*)^2/\delta^2 + 4\ell\Delta^* \OPT/\delta + \OPT^2}\\
  &\le \frac{4 \ell \Delta^* \OPT}{4\ell\Delta^* \OPT/\delta}\\
  &= \delta.
\end{align*}
Therefore, it follows that
\begin{align*}
\Prob{}{\frac{\ell \Delta^*}{\delta} + \frac{\OPT}{2} + \ell g(S^*) \ge \OPT} \ge
  1 - \delta.
\end{align*}

Since \change{for all random outcomes} we have $\OPT' \ge g(S^*)$ and $\OPT \ge \OPT'$,
it follows that 
\begin{align*}
  &\frac{\ell \Delta^*}{\delta} + \ell g(S^*) \ge \frac{\OPT}{2}
  \implies
  \frac{2\ell }{\delta} \parens*{\Delta^* + \OPT'} \ge \OPT \ge \frac{1}{2}\parens*{\Delta^* + \OPT'}.
\end{align*}
Therefore, if we query all marginals to compute $\Delta^*$
and then subsample by $1/\ell$, then with probability at least $1 - \delta$
we know that $\OPT$ lies within an interval of ratio $4\ell/\delta$.
\end{proof}

\reduceRatio*

\begin{proof}
Subsample the ground set $N$ with probability $1/\ell$ to get $N'$.
We decide the value of~$\ell$ later as a function of $R$.
Let $\change{m} = \ceil{\log_{2}(R)}$ and set a smaller
failure probability $\change{\hat\delta_R = \delta/(2(m + 1)\log(R))}$.
For a value $p \in [0, 1)$ that we also set later, run
$\AdaptiveSampling(f,k,\tau,1-p,\change{\hat\delta_R / n})$ on $N'$ in parallel
for the values \change{$\tau_i = 2^i(L/k)$ for $i= 0,1,\dots, m$}.
For each call, we determine if
$\OPT' \le 2k\tau$ or $pk\tau \le \OPT'$
as explained in \Cref{sec:binary-search}.
There are three cases to consider:
\begin{enumerate}
  \item If we have $\OPT' \le 2k\change{\tau_0}$, then $\OPT' \le 2 L$.
  \item If we have $\OPT' \ge pk\change{\tau_m}$, then $\change{\OPT} \in [pU, U]$
    since $\OPT' \le \OPT \le U$.
  \item Otherwise, find the least index \change{$i^*$} such that \change{$\OPT' \ge pk(2^{i^*}L/k)$
    and $\OPT' \le 2k(2^{i^*+1}L/k)$.}
    This implies that \change{ $\OPT' \in [p 2^{i^*} L, 2^{i^*+2}L]$.}
\end{enumerate}
For the first case,
\change{we apply \Cref{lem:subsample_approx}} and
observe that with probability at least \change{ $1 - \hat\delta_R$},
\begin{align*}
  \frac{1}{2}\parens*{\Delta^* + L} &\le \OPT
  \le \frac{2 \ell}{\change{\hat\delta_R}} \parens*{\Delta^* + \OPT'}
  \le \frac{2 \ell}{\change{\hat\delta_R}} \parens*{2\Delta^* + 2L}.
\end{align*}
Therefore, we have a new interval containing $\OPT$ whose
ratio is $4\ell/\change{\hat\delta_R}$.

For the second case,
\change{we have $\OPT \in [pU, U]$,
so the ratio of the new feasible region is $1/p$.}
%it follows from the \Cref{lem:subsample_approx} and the case
%that
%\begin{align*}
%  \frac{p}{2} \parens*{\Delta^* + U} \le \OPT \le \frac{2\ell}{\hat\delta}\parens*{\Delta^* + U}.
%\end{align*}
%Therefore, the ratio of the new feasible region is $4\ell/(\hat\delta p)$. 

\change{For the third case, 
it follows from \Cref{lem:subsample_approx}
and the case assumption that,
with probability at least $1-\hat\delta_R$,
\begin{align*}
  \frac{p}{2} \parens*{\Delta^* + 2^{i^*}L}
  \le
  \OPT
  \le
  \frac{8\ell}{\hat\delta_R}\parens*{\Delta^* + 2^{i^*}L}.
\end{align*}
It follows that the ratio of the new feasible region is $16\ell /(p \hat\delta_R)$.}
Thus, in all cases the ratio $R$ maps to a new ratio
of size at most \change{$16\ell /(p \hat\delta_R)$
with probability at least $1 - 2(m+1)\hat\delta_R = 1 - \delta/\log(R)$ by a union bound.}
\change{Setting $\ell = \log^2(R)$ and $p = 1/\log(R)$,
for any $R \ge 100$,
the new ratio is at most
\[
    \frac{16 \ell}{p \hat\delta_R}
    =
    \frac{32 \log^4(R) (m+1)}{\delta}
    \le
    \frac{64 \log^5(R)}{\delta}.
\]
}

Now we assume $R \ge 100$ and \change{analyze} the adaptivity and
query complexity of the \change{ratio} reduction procedure.
The adaptivity is that of
$\AdaptiveSampling(f,k,\tau,1-p,\change{\hat\delta_R/n})$ because
we try all values of \change{$\tau_i$} in parallel.
Thus, by \Cref{lem:adaptive_sampling}
\change{and our choice of $\hat\delta_R$,}
the adaptivity complexity\footnote{\change{Note that the expected size of the ground after subsampling is $n / \ell$.
To make this upper bound deterministic, we can instead randomly permute the ground set and take the first $\ceil{n / \ell}$ elements.
}}
is
\begin{align*}
  O\parens*{\frac{\log\parens*{\frac{\change{n^2}/\ell}{\change{\hat\delta_R}}}}{\log(1/p)}}
  &=
  O\parens*{\frac{\log\parens*{\frac{n/\ell}{\change{\hat\delta_R}}}}{\log(1/p)}} \\
  &=
  O\parens*{\frac{\log\parens*{\frac{n/\log^2(R)}{\change{\hat\delta_R}}}}{\log\log(R)}} \\
  &=
  O\parens*{\frac{\log\parens*{n/\delta}}{\log\log(R)}}.
\end{align*}
Similarly, the expected number of queries is
\begin{align*}
  O\parens*{\change{m} \cdot \frac{n/\ell}{1 - p}}
  &=
  O\parens*{\log(R) \cdot \frac{n/\ell}{1 - p}} \\
  &=
  O\parens*{\log(R) \cdot \frac{\frac{n}{\log^2(R)}}{1 - \frac{1}{\log(R)}}} \\
  &=
  O\parens*{\frac{n}{\log(R)}},
\end{align*}
which completes the proof.
\end{proof}

% \change{The next two results are useful for helping us find a constant-sized ratio where we can terminate the binary search.}

% \change{
% \begin{fact}
% \label{lem:lambert}
% For any $a \in \R_{> 0}$ and $b \in \Z_{\ge 1}$, the largest real solution to
% \[
%     a \cdot \log^{b}(x) = x
% \]
% is
% \begin{align*}
%     x^*
%     &=
%     \exp\parens*{-b \cdot W_{-1}\parens*{-\frac{1}{b} \cdot \frac{1}{a^{1/b}}}},
% \end{align*}
% where $W_{-1}(z)$ is the analytic continuation of the Lambert $W$ (or product log) function.
% \end{fact}
% }

% \change{
% \begin{theorem}[{\cite[Theorem 1]{chatzigeorgiou2013bounds}}]
% \label{thm:lambert_bounds}
% The Lambert function $W_{-1}(-e^{-u-1})$ for $u > 0$ is bounded as follows
% \[
%     -1 - \sqrt{2u} - u
%     <
%     W_{-1}\parens*{-e^{-u-1}}
%     <
%     -1-\sqrt{2u}-\frac{2}{3}u.
% \]
% \end{theorem}
% }

\linearPreprocess*

\begin{proof}
Assume $\delta > 0$ is a constant failure probability.
\morteza{We start by bounding the number of adaptive rounds. The algorithm \SubsamplePreprocess consists of multiple iterations of the while loop. The first iteration starts with $R = U/L = k$, and we progress to the next iteration with a new value of $R$. We denote the sequence of these values by
$R_0 = k, R_1, \dots, R_{X}$ where $R_X$ is the final ratio of $R = U/L$ for which we complete the iteration of the while loop. In other words, we terminate the while loop after the iteration with $R=R_X$ is complete either on Line~\ref{algline:while-loop} or \ref{algline:loglog-did-not-decrease-enough}.  
Lemma~\ref{lem:reduce_ratio} upper bounds the number of adaptive rounds in iteration $0 \leq j \leq X$ by $O\parens*{\frac{\log\parens*{n/\delta}}{\log\log(R_j)}}$.
Thus, it suffices to upper bound the sum $\sum_{j=0}^X 1/\log\log(R_j)$.
Note that the if condition on Line~\ref{algline:if-condition-loglog-decline} ensures that
\[
    \frac{1}{\log\log(R_j)} \leq \frac{1}{2 \log\log(R_{j+1})}.
\]

The summands form a geometrically increasing series,
so the summation is upper bounded by twice its largest term:
\begin{equation}\label{eq:bound-inverse-loglogR-sum}
    \sum_{j = 0}^X \frac{1}{\log\log(R_j)} \leq \frac{2}{\log\log(R_X)} \leq 1,
\end{equation}
where the last inequality holds because the while condition on Line~\ref{algline:while-loop}
tells us $R_X \ge 2000$.

The expected number of queries can be bounded similarly.
\Cref{lem:reduce_ratio} bounds the expected number of queries in each iteration by $O(n/\log(R))$.
Therefore, we need to upper bound the sum $\sum_{j=0}^X 1/\log(R_j)$ by a constant.
This is evident using \Cref{eq:bound-inverse-loglogR-sum} and the fact that $1/\log(R) \leq 1/\log\log(R)$.

Now let us bound the overall failure probability. In iteration $0 \leq j \leq X$, we set $\hat\delta_R = \delta/(2(m+1) \log(R))$ where $R = R_j$.
We call \AdaptiveSampling $m+1$ times and we may also use the bounds on $\OPT$ stated in Lemma~\ref{lem:subsample_approx} for any of the $m+1$ threshold values $\{\tau_i\}_{i=1}^m$.
We apply the union bound on all of these $2(m+1)$ potential failure events
and conclude that the results of iteration $j$ is reliable (i.e., the run is successful)
with probability at least $1 - \delta/\log(R_j)$. Using the above arguments and \Cref{eq:bound-inverse-loglogR-sum},
the overall failure probability across all iterations is at most:
\[
    \sum_{j=0}^X \frac{\delta}{\log(R_j)} \leq \delta.
\]

Finally, it remains to prove that after running \SubsamplePreprocess, we bound $\OPT$ in the desired constant-ratio range.
Let us define $\alpha = 64/\delta$ and the function $h(R) = \alpha \cdot \log^5(R)$. Using Lemma~\ref{lem:reduce_ratio}, we know that in each iteration we reduce $R$ to some value at most $h(R)$.

If \SubsamplePreprocess terminates after the while condition on Line~\ref{algline:while-loop} fails,
for the final $L$ and $U$, we have $U/L < R^* = 2000$.

Otherwise, we terminate because of the if condition on Line~\ref{algline:if-condition-loglog-decline}.
So, for the final $R$, we have:
\[
\log\log(R) < 2\log\log(h(R)) \implies \log(R) < \log^2(h(R)) = \parens*{\log(\alpha) + 5\log\log(R)}^2.
\]
Therefore,
\[
\log(R) < 4 \max\{\log^2(\alpha), 25\parens*{\log\log(R)}^2\}.
\]
The $\max$ could take each of its two terms values. If the first case occurs, we have $\frac{1}{4}\log(R) < \log^2(\alpha)$. In this case, we have
\[
    \log(R^{1/4}) < \log^2(\alpha)
    \implies
    R^{1/4} < e^{\log^2(\alpha)}
    = \parens*{e^{\log(\alpha)}}^{\log(\alpha)}
    = \alpha^{\log(\alpha)}.
\]
This means $R$ is at most $\alpha^{4\log(\alpha)}$, yielding the upper bound we want. 

If the second case occurs, we have $\log(R) < 100 (\log\log(R))^2$.
The two sides of this inequality do not have the same asymptotic complexity, so the limit of their ratio (RHS divided by LHS) goes to zero as $R$ approaches infinity.
This means there exists a constant $C$ such that for $R > C$, the inequality $\log(R) < 100 (\log\log(R))^2$ cannot hold.
Therefore, $R$ is at most this constant $C$ and consequently $R = O(1)$, which completes the proof.
}
\end{proof}

\section{Missing Analysis from \Cref{sec:submodular-cover}}

\subsection{Analysis of the \WrapperForCover Algorithm}
\label{app:submodular-cover}

\submodularCover*

\begin{proof}
\change{Start by assuming all calls to \EstimateMean give correct
outputs (i.e., event $Z$ occurs).
This happens with probability at least $1-1/n$ by our choice of $\hat\delta$
and a union bound.}
\change{Furthermore, assume}
that $\Delta^* < L$, since if $\Delta^* \ge L$
\change{then the algorithm can trivially output the singleton with the
largest marginal value.}

Next, observe that we have $f(S) \ge L$ upon termination
since we assumed $f$ is integer-valued and
the threshold can eventually reach $\tau < 1$.
To bound the adaptivity complexity, observe that \AdaptiveSamplingForCover
runs in $\change{O(\log(n^2(m+1)))} = O(\log(n\change{(\log(\Delta^*) + 1)}))$ adaptive rounds by
\Cref{cor:threshold-cover} and our choice of $\varepsilon$ and
$\delta$ in Lines~1--3 of \WrapperForCover.
\WrapperForCover calls this subroutine 
$\change{O(m)} = O(\change{\log(\Delta^*) + 1})$ times.
\change{Thus, the}
adaptivity complexity is
$O(\log(n\log(L))\log(L))$ by our initial assumption
\change{$1 \le \Delta^* < L$}.

For the approximation factor of \WrapperForCover,
we begin by mirroring the
analysis of the approximation factor for submodular maximization in
Theorem~\ref{thm:approx_factor}.
Recall the value and marginal gain of the averaged process
$\hat{f}(S_{i})$ and $\hat\Delta(X_i,S_{i-1})$ defined in \Cref{sec:maximization}, and
let $k^* = |S^*|$ denote the size of the optimal set $S^*$.
Call the subsets that are added to~$S$ during the course of the algorithm
$T_1,T_2,\dots,T_m$ and let the remaining gap be $\delta_i = f(S^*) - \hat{f}(S_i)$.
Following the proofs of Lemmas~\ref{lem:conditional_iterative_property}
and \ref{lem:expectation_approx}, and noticing
that $f(S^*)/k^* \le \Delta^*$ by submodularity, for all $i \ge 1$, 
we have
%\[
%  \ExpCond{}{\hat\Delta\parens*{X_i,S_{i-1}}}{H_i(h)}
%  \ge \frac{\parens*{1-\varepsilon}^2}{k^*} \cdot
%  \ExpCond{}{f\parens*{S^*} - \hat{f}\parens*{S_{i-1}}}{H_i(h)}.
%\]
%Letting $\delta_i = f(S^*) - \hat{f}(S_i)$, it follows from the proof
%of \Cref{lem:expectation_approx} that
\begin{align}
\label{eqn:dist_to_opt_drop}
  \ExpCond{}{\delta_i}{H_{i-1}(h)} &\le
  \parens*{1 - \frac{(1-\varepsilon)^2}{k^*}}
  \cdot \ExpCond{}{\delta_{i-1}}{H_{i-1}(h)}.
\end{align}
While this expected inequality holds when conditioned on histories
$h=(T_1,T_2,\dots,T_{\ell-1})$, we show how to iterate it so that the
gap $L-f(S)$ decreases geometrically with constant probability.

We start by showing that the size of every subset $T_i$ is at most
$|T_i| \le k^*/(1-\varepsilon)^2 = 4k^*$. Since $f$ is a monotone submodular function
and $\tau$ is reduced by a factor of $(1-\varepsilon)$ each time
\AdaptiveSamplingForCover is called starting from $\tau=\Delta^*$,
it follows from Property 2 of \Cref{cor:threshold-cover} that
\[
  f(S^*) \le f(S) + \sum_{x \in S^*} \Delta(x,S) \le
  f(S) + k^*  \frac{\tau}{1-\varepsilon}.
\]
It follows from Line~13 in \AdaptiveSamplingForCover that an upper bound for $|T_i|$ is 
\[
  \abs*{T_i}
  \le
  \change{\ceil*{\frac{L - f(S)}{(1-\varepsilon)\tau}}}
  \le
  \change{\ceil*{\frac{f(S^*) - f(S)}{(1-\varepsilon)\tau}}}
  \le
  \change{\ceil*{\frac{k^*}{(1-\varepsilon)^2}}}
  =
  4k^*.
\]

Now we consider the progress of reducing the gap $L - f(S)$ after adding
blocks of sets $T_i$.
Define the first block $B_1 = T_1 \cup T_{2} \cup \dots \cup T_{\ell}$
such that $t_1+t_2+\dots+t_\ell \ge 4k^*$ for the least
possible value of~$\ell$.
Similarly, define the blocks $B_2, B_3, \dots$ to be the union of the sets
$T_i$ after the previous block such that the cardinality first exceeds
$4k^*$.
\change{Since} $|T_i| \le 4k^*$, we have the upper bound $|B_i| \le
8k^*$, which we use to ensure that the algorithm processes
sufficiently many blocks.
Since we analyze the algorithm by blocks, it is convenient to let
$S_{B_i} = \bigcup_{j=1}^i B_j$ denote the union of the first~$i$ blocks.
Lastly, observe that
\[
    \Delta(B_{i}, S_{B_{i-1}})
    \le
    4(f(S^*)-f(S_{B_{i-1}})),
\]
for all $i \ge 1$, because the addition of each block never exceeds the
gap \change{$L - f(S_{B_{i-1}})$}
by a factor of more than $1/(1-\varepsilon)^2$
and $L \le f(S^*)$.

By analyzing the algorithm with blocks of size $O(k^*)$, we show that
the addition of each block independently reduces the current gap
$L - f(S_{B_{i-1}})$ by a constant factor with probability $p \ge 0.05$.
This allows us to analyze the expected output size $\E[|S|]$ via
the negative binomial distribution.
Using an analogous block indexing for the gap $\delta_i$,
%observe that our earlier upper bound on
%$\ExpCond{}{\delta_{i}}{H_{i-1}(h)}$ implies
observe that \Cref{eqn:dist_to_opt_drop}
and \change{$\varepsilon=1/2$ imply that}
\begin{align*}
  \ExpCond{}{\delta_{B_{i}}}{S_{B_{i-1}}}
  &\le \parens*{1 - \frac{\change{(1-\varepsilon)^2}}{k^*}}^{4k^*} \cdot \ExpCond{}{\delta_{B_{i-1}}}{S_{B_{i-1}}} \\
  &\le (1/e) \cdot \ExpCond{}{\delta_{B_{i-1}}}{S_{B_{i-1}}}.
\end{align*}
Since blocks are unions of complete sets $T_j$, the
averaged process $\hat{f}(S_{B_{i}})$ and the true value $f(S_{B_i})$ always agree,
conditioned on the previous state $S_{B_{i-1}}$.
Using the inequality above,
\begin{align*}
  \ExpCond{}{\Delta\parens*{B_i, S_{B_{i-1}}}}{S_{B_{i-1}}}
  &=
  \change{\ExpCond{}{\delta_{B_{i-1}}}{S_{B_{i-1}}}
  -
  \ExpCond{}{\delta_{B_{i}}}{S_{B_{i-1}}}} \\
  &\ge
  \change{\ExpCond{}{\delta_{B_{i-1}}}{S_{B_{i-1}}}
  - (1/e) \cdot \ExpCond{}{\delta_{B_{i-1}}}{S_{B_{i-1}}}}
  \\
  &=
  \parens*{1 - 1/e} \cdot
  \ExpCond{}{\delta_{B_{i-1}}}{S_{B_{i-1}}}.
\end{align*}
This means that the addition of each block $B_i$ decreases the current gap
$L - f(S_{B_{i-1}})$ by a constant factor in expectation.
However, we can make a stronger claim since $\Delta(B_i, S_{B_{i-1}})$ is upper bounded.

Let $X_i$ be the indicator random variable conditioned on $S_{B_{i-1}}$ such that
\[
  X_i = \begin{cases}
    0 & \text{if $\Delta(B_{i},S_{B_{i-1}}) < (1-2/e)\cdot\delta_{B_{i-1}}$,}\\
    1 & \text{otherwise}.\\
  \end{cases}
\]
We claim that $X_i = 1$ with probability $p \ge 0.05$;
otherwise, we would have
\begin{align*}
  \ExpCond{}{\Delta\parens*{B_i, S_{B_{i-1}}}}{S_{B_{i-1}}}
  &< 
  (1-p)(1-2/e)\cdot\ExpCond{}{\delta_{B_{i-1}}}{S_{B_{i-1}}}
   + 4p \cdot \ExpCond{}{\delta_{B_{i-1}}}{S_{B_{i-1}}}\\
  &< (1 - 1/e) \cdot \ExpCond{}{\delta_{B_{i-1}}}{S_{B_{i-1}}},
\end{align*}
which is a contradiction.
Therefore, for each block $B_i$ we have
%with probability $p \ge 0.05$ we independently have
\[
  \ProbCond{}{\delta_{B_{i}} \le (2/e)\cdot\delta_{B_{i-1}}}{S_{B_{i-1}}} \ge 0.05.
\]
In other words, with probability $p \ge 0.05$, the addition of each block
independently decreases the remaining gap to $f(S^*)$ by a constant factor.

Thus, if after the addition of $\ell$ blocks there are
$a=\ceil{\log(L)/\log(e/2)}$ events such that $X_i = 1$,
then the current gap to $f(S^*)$ satisfies
\[
  \delta_{B_\ell} \le (2/e)^a \cdot \delta_{B_0} \le \frac{1}{L}\cdot \delta_{B_0}.
\]
By the definition of $\delta_{B_i}$ and the assumption that $f$ is nonnegative,
this implies that
\begin{align*}
  f\parens*{S^*} - f\parens*{S_{B_\ell}} \le \frac{1}{L}\cdot f\parens*{S^*}
  \implies L\parens*{1-\frac{1}{L}} = L - 1 \le f\parens*{S_{B_\ell}}.
\end{align*}
\change{Since} we assumed that $f$ is integer-valued,
the algorithm reaches the value lower bound $L$ after the addition of the next item.

It follows that we can upper bound $\E[|S|]$ by the expected number of blocks
needed to have~$a$ successful events plus one more block to ensure that we exceed the
target value $L$ (conditioned on all calls to
\AdaptiveSamplingForCover succeeding, which by our choice of $\delta$
happens with probability at least $1 - 1/n$).
Since each block has at most $8 k^*$ elements,
noticing that this stopping criterion is given by the
negative binomial distribution yields
\begin{align*}
  \Exp{}{\abs*{S}} &\le 8k^*\parens*{1 + \sum_{\ell=0}^\infty \parens*{\ell+a}\binom{\ell+a-1}{\ell}(1-p)^\ell p^a}\\
  &= 8k^* \parens*{1 + a + \sum_{\ell=0}^\infty \ell \binom{\ell+a-1}{\ell}(1-p)^\ell p^a}\\
  &= 8k^* \parens*{1 + a + \frac{(1-p)a}{p}}\\
  &\le 8k^*\parens*{20a + 1}.
\end{align*}
Here we use the fact that the expected value of a negative binomial distribution
parameterized by~$a$ successes and failure probability $1-p$
is $(1-p)a/p$.
Since $a = O(\log(L))$,
it follows that we have the conditional expectation
$\E[|S|] = O(k^* \log(L))$ with probability at least $1 - 1/n$.
Conditioned on the algorithm failing (which happens with probability at most $1/n$),
we have $|S| \le n$.
Therefore, in total we have
\begin{align*}
  \Exp{}{|S|} &= \parens*{1 - 1/n} \cdot O(k^* a) + (1/n)\cdot n
  = O(k^* \log(L)),
\end{align*}
as desired. This completes the proof of the expected approximation factor.
\end{proof}

\end{document}